\documentclass[format=acmsmall,screen,nonacm]{acmart}
\usepackage{acm-ec-22}
\usepackage{booktabs} 
\usepackage[ruled]{algorithm2e} 

\SetAlFnt{\small}
\SetAlCapFnt{\small}
\SetAlCapNameFnt{\small}
\SetAlCapHSkip{0pt}
\IncMargin{-\parindent}

\setcitestyle{acmnumeric}

\definecolor{ForestGreen}{rgb}{0.0, 0.27, 0.13}

\usepackage[utf8]{inputenc}
\usepackage{verbatim}

\title{An Algorithmic Solution to the Blotto Game using Multi-marginal Couplings}
\author{Vianney Perchet}
\affiliation{Crest, ENSAE}
\affiliation{Criteo AI LAB}
\email{vianney.perchet@normalesup.org}
\author{Philippe Rigollet}
\affiliation{Department of Mathematics, MIT}
\email{rigollet@math.mit.edu}
\author{Thibaut Le Gouic}
\affiliation{Aix Marseille University, Centrale Marseille, CNRS, I2M}
\email{thibaut.le_gouic@math.cnrs.fr}
\authorsaddresses{vianney.perchet@normalesup.org, rigollet@math.mit.edu, thibaut.le\_gouic@math.cnrs.fr}
\date{\today}
\usepackage{hhline}
\usepackage{natbib}
\usepackage{graphicx}
\usepackage{dsfont}
\usepackage{hyperref}
\usepackage{algpseudocode}

\usepackage{pifont}
\newcommand{\xmark}{\ding{55}}

\hypersetup{
    colorlinks=true,
    linkcolor=blue,
    filecolor=magenta,      
    urlcolor=cyan,
}
\usepackage{url}
\usepackage[normalem]{ulem}

\usepackage{subcaption}
\usepackage{tikz}
\usetikzlibrary{3d}
\usetikzlibrary{arrows.meta}

\newcommand{\R}{\mathds{R}}

\newcommand{\E}{\mathds{E}}
\newcommand{\p}{\mathds{P}}
\newcommand{\1}{\mathds{1}}

\newcommand{\cN}{\mathcal{N}}
\newcommand{\cO}{\mathcal{O}}
\newcommand{\cH}{\mathcal{H}}
\newcommand{\cC}{\mathcal{C}}
\newcommand{\cF}{\mathcal{F}}
\newcommand{\cI}{\mathcal{I}}
\newcommand{\cU}{\mathcal{U}}
\newcommand{\cG}{\mathcal{G}}

\newcommand{\bp}{\mathbf{p}}

\newcommand{\sH}{\mathsf{H}}

\newcommand{\supp}{{\rm supp }}

\DeclareMathOperator{\var}{\mathrm{var}}
\newcommand{\eps}{\varepsilon}

\newcommand{\Cube}{\textsc{Extremize}}
\newcommand{\SAMPLE}{\textsc{Sample}}
\newcommand{\DISCRETIZE}{\textsc{Discretize}}
\newcommand{\SINKHORN}{\textsc{Sinkhorn}}
\newcommand{\REDUCTION}{\textsc{Reduc}}
\newcommand{\DECOMPOSE}{\textsc{Decomp}}
\newcommand{\LOTTO}{\textsc{Lotto}}
\newcommand{\LOTTOTOBLOTTO}{\textsc{Lotto2Blotto}}

\newtheorem{theorem}{Theorem}
\newtheorem{definition}[theorem]{Definition}

\newtheorem{remark}[theorem]{Remark}

\newtheorem{lemma}[theorem]{Lemma}
\newtheorem{proposition}[theorem]{Proposition}
\newtheorem{corollary}[theorem]{Corollary}

\DeclareMathOperator{\KL}{\mathsf{KL}}

\DeclareMathOperator*{\argmax}{\mathrm{argmax}}

\begin{abstract}We describe an efficient algorithm to compute solutions for the general two-player Blotto game on $n$ battlefields with heterogeneous values. While explicit constructions for such solutions have been limited to specific, largely symmetric or homogeneous, setups, this algorithmic resolution covers the most general situation to date: value-asymmetric game with asymmetric budget with sufficient symmetry and homogeneity. The proposed algorithm rests on recent theoretical advances regarding Sinkhorn iterations for matrix and tensor scaling. An important case which had been out of reach of previous attempts is that of heterogeneous but symmetric battlefield values with asymmetric budget. In this case, the Blotto game is constant-sum so optimal solutions exist, and our algorithm samples from an $\eps$-optimal solution in time $\tilde \cO(n^2 + \eps^{-4})$, independently of budgets and battlefield values. In the case of asymmetric values where optimal solutions need not exist but Nash equilibria do, our algorithm samples from an $\eps$-Nash equilibrium with similar complexity but where implicit constants depend on various parameters of the game such as battlefield values.\end{abstract}

\begin{document}

\begin{titlepage}

\maketitle

\end{titlepage}

\section{Introduction}

A century ago, Emile Borel published his seminal paper on \emph{the theory of play and integral equations with skew symmetric kernels}~\cite{Bor21}, see also \cite[Page 157]{von1959role}. While perhaps not as conspicuous, it predates von Neumann's monumental work \emph{on the theory of games of strategy}~\cite{von28} by several years. In this work, Borel describes what is now called the Blotto game and introduces the notions of strategy, mixed strategies and even foresees the fruitful interactions between game theory and economics that are to be observed throughout the century. As such, the Blotto game is considered to be the genesis of modern game theory~\cite{Fre53, Nak06}. Despite its prestigious pedigree, equilibrium strategies for this game are only known in special cases.

Blotto is a resource-allocation game in which two players competes over $n$ different battlefields by simultaneously allocating resources to each battlefield. The following two additional characteristics are perhaps the most salient features of the Blotto game:
\begin{enumerate}
    \item \emph{Winner-takes-all}: For each battlefield, the player allocating the most resources to a given battlefield wins the battlefield.
    \item \emph{Fixed budget}: each player is subject to a fixed---and deterministic---budget that mixed strategies should satisfy almost surely.
\end{enumerate}

Despite its apparent simplicity the Blotto game captures a variety of practical situations that extend far beyond the context of the above military terminology. These include political strategy \cite{Mye93,LasPic02, MerMunTof05}, network security \cite{LabHaSaa15, FerSaaMan21}, and various forms of practical auction markets \cite{MasSil15,HajMan17}.

The goal of this paper is to efficiently construct a Nash equilibrium for this game or, when they exist, an optimal strategy.

\medskip

\noindent{\bf Prior work.} Despite its century-long existence, Nash equilibria for the Blotto game are only known under various restrictions on the main parameters of the problem: the budget of each player and the value given to each battlefield.
\begin{itemize}
    \item \emph{Budget.} A large fraction of the literature considers the case where the players have \emph{symmetric budgets}, starting with the original problem of Borel~\cite{Bor21} and in most of the main contributions throughout the twentieth century~\cite{Bor21, BorVil38,Gro50, GroWag50, Las02,Tho18}. The case of symmetric budgets is well understood except in the setup where players may disagree on the value of battlefield that was recently introduced~\cite{KovRob21}.
    \item \emph{Battlefields.} When the two players have a different budget the situation becomes more complex as the poorest will have to forfeit some battlefields. In this case, only partial results are known. To understand what ``partial" means, recall that full generality of the battlefield values occurs when (i) players may assign a different value to a given battlefield---we say that the values are \emph{asymmetric}---and (ii) these values may vary across battlefield---we say that the battlefields are \emph{heterogeneous}.  Partial results are known for symmetric values. Even under this simplifying assumption, the case of heterogeneous battlefields remains poorly understood, except in the case of two battlefields~\cite{MacMas15}. In the case of more than two battlefields, Nash equilibria are known for homogeneous battlefields~\cite{Rob06} or under stringent assumptions  on the battlefield values~\cite{SchLoiSas14} that essentially reduce to the homogeneous case. 
\end{itemize}
We refer the reader to Table~\ref{fig:SOTA} for a survey of recent advances. While we tackle the most general setup to date, we stress that an important case was not covered by prior literature: the case of asymmetric budget, heterogeneous and symmetric values. Indeed, in this case, the game is constant-sum and optimal strategies exist. Our results also cover the case of asymmetric values introduced very recently in~\cite{KovRob21} but this setup leads to only Nash equilibria rather than optimal strategies. Whenever possible, we conflate the two setups and simply refer to a \emph{solution} to the Blotto.

A discrete version of the Blotto game where both budgets and allocations are required to be integral was also introduced in Borel's original paper~\cite{Bor21}. Explicit optimal solutions were provided in~\cite{hart2008discrete} for the homogeneous and symmetric version of the discrete game;  see also \cite{hortala2012pure} for partial solutions in the asymmetric-value case. More recently, this discrete version has seen significant computational advances~\cite{AhmDehHaj19, beaglehole2022efficient}. Conceptually, this line of work is close to the present paper in the sense that it provides an algorithm to sample from approximate solutions. Moreover, the discrete Blotto game can be seen as a discretization of the continuous version of interest here and that could be quantified using the arguments of Section~\ref{SE:Errors}. However carrying out this analysis, for instance based on~\cite[Theorem~4.2]{beaglehole2022efficient}, leads to worse dependence on $n$ and $\eps$ compared to Theorem~\ref{TH:complex_cstsum} here. More strikingly, the complexity bound of Theorem~\ref{TH:complex_cstsum} does not depend on budgets or battlefield values while this dependence is polynomial in the bounds for discrete Blotto. 
The two lines of work also differ in more profound ways. First and foremost, the approach employed here is fundamentally different: it aims at mixing known solutions for the related Lotto game while solutions to the discrete Blotto games are more agnostic so that it is unclear what the marginals of the resulting strategy are. In particular, the present approach allows us sample from $\eps$-Nash equilibria in the asymmetric-value case whereas this setup is currently out of reach for solutions to the discrete Blotto game.

Finally, note that our approach also yields new (existential) results for the discrete Blotto game. Since they are not the focus of our contribution, they are relegated to the appendix.

\medskip

\noindent {\bf Our contributions.}
All of the above solutions for two-player games have consisted in constructing explicit solutions. Because of the budget constraints, these strategies can be decomposed in two parts: \emph{marginal distributions} that indicate which (random) strategy to play on each battlefield and a \emph{coupling} that correlates the marginal strategies in such a way to ensure that the budget constrained is satisfied almost surely. 

The first question may be studied independently of the second by considering what is known as the (General) \emph{Lotto} game~\cite{BelCov80}. In this game the budget constraint need only be enforced in \emph{expectation} with respect to the randomization of the mixed strategies. While this setup lacks a defining characteristic of the Blotto game (fixed budget), it has the advantage of landing itself to more amenable computations. Indeed, unlike the Blotto game, a complete solution to the Lotto game was recently proposed in~\cite{KovRob21} where the authors describe an explicit Nash equilibrium in the most general case: asymmetric budget, asymmetric and heterogeneous values.

In light of this progress a natural question is whether the marginal solutions  discovered in~\cite{KovRob21} can be \emph{coupled} in such a way that the budget constraint is satisfied almost surely. We provide a positive answer to this question by appealing to an existing result from the theory of \emph{joint mixability}~\cite{wangJointMixability2016}. Mixability asks the following question: Can $n$ random variables $X_1, \ldots, X_n$ with prescribed marginal distributions $X_i\sim  P_i$, be coupled in such a way that $\var(X_1+ \cdots + X_n)=0$. Joint mixability is precisely the step required to go from a Lotto solution to a Blotto one by coupling the marginals of the Lotto solution in such a way that the budget constraint is satisfied.

In this paper we exploit a simple and new connection between joint mixability and the theory of \emph{multi-marginal couplings} that has recently received a regain of interest in the context of optimal transport~\cite{AguCar11,DiMGerNen17,AltBoi20}. In multi-marginal optimal transport, the goal is to optimize a cost over the space of couplings with given marginals. Unlike the case of two marginals that arises in traditional optimal transport, this question raises significant computational challenges and often leads to  NP-hardness~\cite{AltBoi21}. 
In the language of optimization, joint mixability merely asks if the set of constraints is nonempty. We propose an algorithmic solution to the Blotto problem by efficiently constructing a coupling that satisfies the budget constraint almost surely and can be easily sampled from. Our construction relies on three key steps: first we reduce the problem to a small number of marginals to bypass the inherent NP hardness of multi-marginal problems, second we discretize the marginals and finally, we employ a multi-marginal version of the Sinkhorn algorithm~\cite{Sin64,SinKno67} to construct a coupling of the discretized marginals. After a simple smoothing step, we produce a sampling with continuous marginals that are close to the ones prescribed by the Lotto solutions and from which it is straightforward to sample. Furthermore, we quantify the combined effect of discretization error and of the Sinkhorn algorithm on the value of the game,  effectively leading to an \emph{approximate} Nash equilibrium and even to an approximately optimal solution in the case of symmetric values.

We exhibit tight---or near-tight in the asymmetric values case---conditions for the mixability of \emph{specific} Lotto solutions into Blotto solutions; see Corollaries~\ref{cor:symvals} and~\ref{cor:stability} below. While these conditions are reasonable and cover most cases, some heavily skewed games, either in terms of budget asymmetry or values inhomogeneity, are not covered by our results. We leave it as an open question to exhibit Lotto strategies that can be mixed into Blotto ones even for such games.

The rest of this paper is organized as follows. In the next section, we recall the solution for the Lotto game and show that they can be turned into solutions for the Blotto game. This existential result simply appeals to existing results of joint mixability. We move from an existential to an algorithmic result in Section~\ref{SE:Reductions} by proceeding in three steps: first we reduce the problem to the case $n=4$, then we discretize the problem and finally we apply Sinkhorn algorithm to couple the resulting marginals in a appropriate fashion. The main product of Section~\ref{SE:Reductions} is Algorithm~\ref{Algo:Sample} which shows how to sample from an approximate solution to the Blotto game. Finally, we provide a detailed complexity analysis for this algorithm in Section~\ref{SE:Errors}, showing in particular, that it runs in time polynomial in the parameters of the Blotto game and the approximation error $\eps$. Finally, our techniques also yield new results for the discrete Blotto game largely studied by~\cite{hart2008discrete,hortala2012pure} that are of independent interest. We postpone them to the appendix.

\medskip

\noindent{\sc Notation.} For any integer $n$, define $[n]=\{1, \ldots, n\}$. We use ${\bf 1}$ to denote an all-ones vector or tensor. Note that the dimension of this vector will be clear from the context but may vary across occurrences.  For any two vectors $x,y$, we denote their entrywise (Hadamard) product $x\odot y$ and their entrywise division $x \oslash y$ whenever $y$ has only nonzero entries.
For any two real numbers $a,b$ we denote by $a\vee b$ their maximum and by $a\wedge b$ their minimum.

\section{Solutions for Blotto and Lotto games} \label{SE:Model}

The goal of this section is to describe the Blotto game and its connection to the Lotto game for which explicit solutions are known. We first recall a solution for the Lotto game derived in~\cite{KovRob21} and show that it can be readily turned into a Blotto solution  using the theory of joint mixability.

\subsection{The Blotto game}
The classical two-player \emph{Blotto} game is formalized as follows. Two players, respectively denoted by   $A$ and $B$, are competing over $n \ge 2$ battlefields denoted by $i \in [n]$. Since we focus on two-player games where both players obey the same rules,  it will be convenient when describing the game to denote by $P \in \{A,B\}$ either player and by $\bar P$ the other player so that $(P,\bar P)\in \{(A,B), (B,A)\}$.

The \emph{datum} of a Blotto game is as follows. Player $P \in \{A,B\}$  has a total budget of  $T_P>0$ to allocate across the $n$ battlefields. Moreover, she valuates battlefield $i \in [n]$ to $v_{P,i}>0$ which may differ from $v_{\bar P,i}$. Without loss of generality,  we  assume that $T_A \ge T_B$ to break symmetry and that 
\begin{equation}
    \label{EQ:sumvi}
    \sum_{i \in [n]} v_{A,i} = \sum_{i \in [n]} v_{B,i} =1 
\end{equation}
Indeed, multiplying  the value of all battlefields by the sames' constant has no impact on the players' strategies.

The \emph{rules} of the Blotto game are as follows. 
A {pure strategy}  for player $P$ is an allocation vector $x_P =(x_{P,1}, \ldots, x_{P,n})$ where $x_{P,i}>0$ is the amount allocated to battlefield $i \in [n]$. A mixed strategy for player $P$ is a probability distribution over pure strategies. A salient feature of the Blotto game is that a player $P$ is constrained to playing strategies that satisfy the budget  constraint:  $x_{P,1}+\ldots x_{P,n} \le T_P$. In turn, admissible mixed strategies for the Blotto games are random vectors $X_P \in \R^n$ such that
\begin{equation}
    \label{EQ:Blotto-budget}
    \sum_{i=1}^n X_{P,i}\le T_P \quad \text{almost surely.}
\end{equation}

Given two pure strategies $x_P$ and $x_{\bar P}$ for players $P$ and $\bar P$ respectively, player $P$  wins battlefield $i \in [n]$ if $x_{P,i}>x_{\bar P,i}$  and receives a reward $v_{P,i}>0$. Ties $x_{P,i}=x_{\bar P,i}$ are broken arbitrarily as they are are irrelevant for our analysis.

The existence of Nash equilibria is a consequence of standard game theoretic arguments~\cite{Ren99}. Unfortunately, these general results say little about the structure of equilibrium strategies. At the end of this section, we make partial progress towards this question by describing the marginals of such equilibrium strategies. However, these remain existential results in essence.

This is in stark contrast with the associated Lotto game, described in the following section, where the hard budget constraint is dropped in favor of a constraint in expectation, and whose explicit solutions have been computed.

\subsection{The associated Lotto game}
\label{SE:Lotto}

A \emph{Lotto} game  has the same data and rules as its associated Blotto game except for the almost sure budget constraint~\eqref{EQ:Blotto-budget} which is relaxed to the following \emph{expected} budget constraint:
\begin{equation}
\label{EQ:const_lotto}
     \sum_{i=1}^n \E[X_{P,i}]\le T_P
\end{equation}

This relaxation greatly simplifies the game. In fact, Kovenock and Roberson~\cite{KovRob21} have recently elicited an  explicit characterization of a non-trivial Nash equilibrium  for the most general version of the Lotto game to date; see Table~\ref{fig:SOTA}. In the rest of Section~\ref{SE:Lotto}, we describe their solution in details since it is the basis for ours.

\begin{table}
    \centering
\begin{tabular}{l|c|c|c|c|c|c|}
\cline{2-7}
    \multicolumn{1}{c|}{} &  Continuous & Asymm.\  & Heterogeneous  & Asymm.\ & More than 3 & Complete \\    
        \multicolumn{1}{c|}{} &  Strategy &  Budget &  Values &   Values & Battlefields& Results\\
          \hhline{~======}
     \cite{BorVil38}&{\color{ForestGreen}\checkmark}
&{\color{red}\xmark}&{\color{red}\xmark}&{\color{red}\xmark}&{\color{red}\xmark}&{\color{ForestGreen}\checkmark}\\
     \cite{GroWag50}&{\color{ForestGreen}\checkmark}&{\color{red}\xmark}&{\color{red}\xmark}&{\color{red}\xmark}&{\color{ForestGreen}\checkmark}&{\color{ForestGreen}\checkmark}\\
          \cite{weinstein2012two}&{\color{ForestGreen}\checkmark}&{\color{red}\xmark}&{\color{red}\xmark}&{\color{red}\xmark}&{\color{ForestGreen}\checkmark}&{\color{ForestGreen}\checkmark}
\\

     \cite{Gro50}&{\color{ForestGreen}\checkmark}&{\color{red}\xmark}&{\color{ForestGreen}\checkmark}&{\color{red}\xmark}&{\color{ForestGreen}\checkmark}&{\color{ForestGreen}\checkmark}\\
     \cite{Las02}&{\color{ForestGreen}\checkmark}&{\color{red}\xmark}&{\color{ForestGreen}\checkmark}&{\color{red}\xmark}&{\color{ForestGreen}\checkmark}&{\color{ForestGreen}\checkmark}\\
          \cite{Tho18}&{\color{ForestGreen}\checkmark}&{\color{red}\xmark}&{\color{ForestGreen}\checkmark}&{\color{red}\xmark}&{\color{ForestGreen}\checkmark}&{\color{ForestGreen}\checkmark}
\\
    \cite{GroWag50}&{\color{ForestGreen}\checkmark}&{\color{ForestGreen}\checkmark}&{\color{ForestGreen}\checkmark}&{\color{red}\xmark}&{\color{red}\xmark}&{\color{ForestGreen}\checkmark}\\
     \cite{Rob06}&{\color{ForestGreen}\checkmark}&{\color{ForestGreen}\checkmark}&{\color{red}\xmark}&{\color{red}\xmark}&{\color{ForestGreen}\checkmark}&{\color{ForestGreen}\checkmark}\\
          \cite{MacMas15}&{\color{ForestGreen}\checkmark}&{\color{ForestGreen}\checkmark}&{\color{ForestGreen}\checkmark}&{\color{red}\xmark}&{\color{red}\xmark}&{\color{ForestGreen}\checkmark}
\\
     \cite{SchLoiSas14}&{\color{ForestGreen}\checkmark}&{\color{ForestGreen}\checkmark}&{\color{ForestGreen}\checkmark}&{\color{red}\xmark}&{\color{ForestGreen}\checkmark}&{\color{red}\xmark}
\\
     \cite{KovRob21}&{\color{ForestGreen}\checkmark}&{\color{ForestGreen}\checkmark}&{\color{ForestGreen}\checkmark}&{\color{ForestGreen}\checkmark}&{\color{ForestGreen}\checkmark}&{\color{red}\xmark}
\\
     \hhline{~======}
     This paper & {\color{ForestGreen}\checkmark} &{\color{ForestGreen}\checkmark}&{\color{ForestGreen}\checkmark}&{\color{ForestGreen}\checkmark}&{\color{ForestGreen}\checkmark}&{\color{ForestGreen}\checkmark}\\
     \cline{2-7}
\end{tabular}   
    \caption{Variants of the continuous Blotto game and their solutions. The last column, ``complete results'' indicates whether results obtained hold with possibly strong assumptions on the different values (for instance, there always exist more than 3 battlefields with the exact same value). }
    \label{fig:SOTA}
\end{table}
Finding an optimal strategy for the Lotto game amounts to finding a stationary point for an  optimization problem subject to constraints of the form~\eqref{EQ:const_lotto}. Because of linearity of expectation, the associated Lagrangian is decomposed as the sum of $n$ terms, one per battlefield, that are each mathematically equivalent to an ``all-pay'' auction whose solutions  are well known.

More explicitly,  Nash equilibria of the Lotto problem depend on two parameters $\gamma\geq 0$ and $\lambda\geq 0$, that are set later on. First, given any $\gamma \geq 0$, consider the subsets of battlefields $\cN(\gamma)$ that are  at least $\gamma$-times more valuable to  $A$ than to $B$:
\[
\mathcal{N}(\gamma) = \{ i \in [n], \ \frac{v_{A,i}}{v_{B,i}} \geq \gamma \}\,.
\]

Given a scaling parameter  $\lambda \geq 0$ to be defined later, the mixed strategy of player $A$ at equilibrium prescribes to allocate a (random) budget of $X_{A,i}$ to battlefield $i$ with distribution given by:

\def\arraystretch{2.2}
$$
X_{A,i}\sim \left\{\begin{array}{ll}
\displaystyle {\sf Unif}\left[0,  \frac{\gamma v_{B,i}}{\lambda}\right] & \text{if } i \in \cN(\gamma)\\
\displaystyle \left(1-\frac{v_{A,i}}{\gamma v_{B,i}}\right) \delta_0 + \frac{v_{A,i}}{\gamma v_{B,i}} {\sf Unif}\left[0,  \frac{v_{A,i}}{\lambda}\right] & \text{if } i \notin \cN(\gamma)\,.
\end{array}\right.\,,
$$
where $\delta_0$ denotes the Dirac point mass at $0$. 

The strategy of player $B$ is given by
$$
X_{B,i}\sim \left\{\begin{array}{ll}
\displaystyle \left(1-\frac{\gamma v_{B,i}}{v_{A,i}}\right) \delta_0 + \frac{\gamma v_{B,i}}{ v_{A,i}} {\sf Unif}\left[0, \frac{ \gamma v_{B,i}}{\lambda}\right] & \text{if } i \in \cN(\gamma)\,,\\
\displaystyle {\sf Unif}\left[0,  \frac{v_{A,i}}{\lambda}\right] & \text{if } i \notin \cN(\gamma)
\end{array}\right.\,.
$$
Note that the strategy of $A$ and $B$ are the same except that the roles of $v_{A,i}$ and $\gamma v_{B,i}$ are switched. In that sense, $\gamma$ plays the role of an ``exchange" rate that accounts for discrepancies between budgets and valuations across the two players.
\def\arraystretch{1}

It remains to find the parameters $\gamma$ and $\lambda$ using the budget constraints. For this set of strategies, saturating the total budget constraint~\eqref{EQ:Blotto-budget} readily yields the following two equations:
\begin{align}
\lambda T_A & = \sum_{i \in \mathcal{N}(\gamma)}\frac{\gamma v_{B,i}}{2} +\sum_{i \not\in \mathcal{N}(\gamma)} \frac{(v_{A,i})^2}{2\gamma v_{B,i}}=\frac12\sum_{i=1}^n(\gamma v_{B,i})\wedge \frac{(v_{A,i})^2}{\gamma v_{B,i}}\,, \label{EQ:Lambda1}\\
\lambda T_B & =  \sum_{i \in \mathcal{N}(\gamma)} \frac{(\gamma v_{B,i})^2}{2v_{A,i}} + \sum_{i \not\in \mathcal{N}(\gamma)}\frac{ v_{A,i}}{2}=\frac12\sum_{i=1}^n\frac{(\gamma v_{B,i})^2}{v_{A,i}}\wedge v_{A,i} \,.\label{EQ:Lambda2}
\end{align}

Any pair $(\gamma, \lambda)$ solving the above system of two equations yields a Nash equilibrium. It remains to show that such solutions may be computed efficiently. Observe that eliminating $\lambda$ from the equations yields the following nonlinear equation in $\gamma$:

\begin{align}
    f(\gamma):&=\gamma^3\left(T_A \sum_{i \in \mathcal{N}(\gamma)}\frac{v_{B,i}^2}{v_{A,i}}\right)-\gamma^2T_B\sum_{i \in \mathcal{N}(\gamma)} v_{B,i}+ \gamma T_A\sum_{i \not\in \mathcal{N}(\gamma)}v_{A,i}-T_B\sum_{i \not\in \mathcal{N}(\gamma)} \frac{v_{A,i}^2}{v_{B,i}}\nonumber\\
    &= \gamma T_A\sum_{i=1}^n \left(\gamma^2\frac{v_{B,i}^2}{v_{A,i}}\wedge v_{A,i}\right)  -T_B\sum_{i=1}^n \frac{v_{A,i}}{v_{B,i}}\left(\gamma^2\frac{v_{B,i}^2}{v_{A,i}}\wedge v_{A,i}\right)=0     \label{EQ:defGamma}
\end{align}

Any solution $\gamma^*$ to this equation readily yields a unique $\lambda^*$ by plugging it into either~\eqref{EQ:Lambda1} or~\eqref{EQ:Lambda2}; both equations will yield the same solution by~\eqref{EQ:defGamma}. In turn, the existence and  efficient computation of solutions $\gamma^*$ to~\eqref{EQ:defGamma} are ensured by the following proposition. The bounds on $\gamma^*$ presented in the following proposition depend on the distance between the vectors of battlefield values $v_A:=(v_{A,1}, \ldots, v_{A,n})$ and $v_B:=(v_{B,1}, \ldots, v_{B,n})$. Interestingly the natural measure of distance that emerges is the $\chi^2$-divergence that commonly arises in information theory and statistics; see e.g.~\cite{PolWu22}. The $\chi^2$-divergence $\chi^2(u\|v)$  between two probability vectors $u=(u_1, \ldots, u_n)$ and $v=(v_1, \ldots, v_n)$ is defined by
    $$
    \chi^2(u\|v)=\sum_{i=1}^n \frac{u_i^2}{v_i}-1=\sum_{i=1}^n \left(\frac{u_i}{v_i}-1\right)^2v_i\,.
    $$
It is clear that $\chi^2(u\|v)\ge 0$ with equality if and only if $u=v$.

 \begin{proposition}\label{PR:Gamma} Equation \eqref{EQ:defGamma} has the following properties:
\begin{enumerate}
    \item It always has at least one and at most $3n+3$ solutions $\gamma^*$.
    \item Any solution $\gamma^*$  satisfies
    $$ \frac{T_B}{T_A} \frac{1}{1+\chi^2(v_B \|v_A)}\leq \gamma^* \leq  \frac{T_B}{T_A}(1+\chi^2(v_A\|v_B))
    $$
    \item Computing all solutions can be done in $\mathcal{O}(n\log n)$ operations.
\end{enumerate}
 \end{proposition}
The proof is  based on standard computations, hence postponed to Appendix \ref{SE:ProofOfPR:Gamma}, with the associated Algorithm \ref{Alg:Lotto}.
 
 \begin{remark}\label{rem:symvals}
In case of symmetric values, that is when $v_{A,i}=v_{B,i}=v_i$, the game is constant-sum and each player has a then unique\footnote{In the case of the Lotto game, it is natural to call a strategy an equivalence class of strategies with the same marginals.} optimal strategy given by a unique pair $(\gamma^*, \lambda^*)$~\cite{KovRob21}. In fact, in that case, the unique $(\gamma^*, \lambda^*)$  can be computed analytically as $\gamma^* = T_B/T_A$ and  $\lambda^* = T_B/(2T_A^2)$; this can be easily seen from Proposition~\ref{PR:Gamma} (point 2.), since $\chi^2(v_A\|v_B)=\chi^2(v_B\|v_A)=0$. With these parameters, the optimal strategy of player $A$ is to choose $X_{A,i}$ uniformly at random on $[0, 2T_Av_i]$ and that of player $B$ is to forfeit each battlefield with probability $1-\frac{T_B}{T_A}$ and, to choose $X_{B,i}$ uniformly at random on $[0, 2T_Av_i]$ on battlefield $i$ if not forfeited.
 \end{remark}

 \subsection{From Lotto to Blotto}
In the previous section, we described how to compute solutions of a Lotto game. To turn a strategy for the Lotto game into a strategy for the Blotto game, one can \emph{couple} the marginal strategies of a Lotto game, effectively turning the constraint~\eqref{EQ:const_lotto} on the expected budget into the almost sure budget constraint~\eqref{EQ:Blotto-budget}.

Stated otherwise, a solution to the associated Lotto game induces a solution to the original Blotto game if the random variables $\{X_{A,i}\}_{i \in [n]}$ (and similarly $\{X_{B,i}\}_i$) are \emph{jointly mixable}~\cite{wangJointMixability2016}.

\begin{definition}
A family of $k$ random variables $Z_1, \ldots, Z_k$ with finite expectations is \emph{jointly mixable} if there exists a coupling $\pi$ such that if $(Z_1, \ldots, Z_k)\sim \pi$,
$$
\sum_{i=1}^k Z_i = \sum_{i=1}^k \E[Z_i], \quad \text{almost surely.}
$$
In that case, the coupling $\pi$ is called a \emph{joint mix}.
\end{definition}
Obviously, not all random $k$-tuples variables are jointly mixable. Take for example $Z_1, Z_2$ and $Z_3$ to be Bernoulli with parameter $1/2$. Then $\E[Z_1]+\E[Z_2]+\E[Z_3]=3/2$ whereas there is no coupling of the $Z_i$s such that their some equals a fractional number.

While the full characterization of jointly mixable distribution is a complex question, some  conditions, either sufficient or necessary, for joint mixability have been derived. The following proposition is a simple extension of a result of \cite{wangJointMixability2016} (see also,~\cite{Zim20}) on the mixability of distributions with monotone densities.
\begin{proposition}\label{PR:JointMixability}
For $i=1, \ldots, k$, let $p_i\in [0,1]$, $b_i>0$ be fixed parameters and let $Z_i$ be a random variable with distribution given by the following mixture:
\begin{equation}
    \label{EQ:Zimixt}
Z_i\sim (1-p_i) \delta_0 + p_i{\sf Unif}\left[0, b_i \right]\,.
\end{equation}
Then $Z_1,\ldots,Z_k$ are jointly mixable if and only if
\begin{equation}
    \label{EQ:JMcond1}
    \max_{1\le i \le k} b_i\le \frac12\sum_{i=1}^k p_ib_i
\end{equation}
\end{proposition}
This proposition is a consequence of few computations; its proof is delayed to Section~\ref{SE:ProofJointMixability}.

We are now in a position to state the main result of this section: the marginal distributions of the Lotto game described above are jointly mixable into a solution to the Blotto game. To that end, we instantiate Proposition~\ref{PR:JointMixability} to the parameters of the marginal distributions described in Section~\ref{SE:Lotto}.

\begin{theorem}\label{TH:LottotoBlotto}
Let $\gamma^*,\lambda^*$ be the parameters of Nash equilibrium of the Lotto game described 
in Section~\eqref{SE:Lotto}. Then the marginal distributions can be coupled into a Nash equilibrium for the corresponding Blotto game if and only if
\begin{equation}
    \label{EQ:blotto_cond_mix}
     \max_{i \in [n]} (\gamma^* v_{B,i} \wedge v_{A,i}) \le \lambda^*T_B\,.
\end{equation}
\end{theorem}
Inequality~\eqref{EQ:blotto_cond_mix} is simply an instantiation of~\eqref{EQ:JMcond1}, hence details are postponed to Section~\ref{SE:ProofJointMixability}

Condition~\eqref{EQ:blotto_cond_mix} of the previous theorem relies on the values $\gamma^*, \lambda^*$ that define the solution of the Lotto game. In light of the bounds obtained in Proposition~\ref{PR:Gamma}, these parameters may be eliminated to produce a sufficient condition for the existence of said solution for Blotto games with symmetric values. Recall that in this case, the game is constant-sum so a solution is, in fact, an optimal strategy. This result is captured in the following corollary which is  a straightforward consequence of Theorem~\ref{TH:LottotoBlotto}.

 \begin{corollary}
 \label{cor:symvals}
Assume symmetric values: $v_{A}=v_{B}=v$. Then the marginal distributions of the optimal Lotto strategy described in Section~\ref{SE:Lotto} with $\gamma^*=T_B/T_A$ and $\lambda^*=T_B/(2T_A^2)$ can be coupled into an optimal strategy for the corresponding Blotto game if and only if
 $$
 \max_{i \in [n]} v_i \le \frac{T_B}{2T_A}\,.
 $$
 \end{corollary}

In fact a sufficient condition may be derived in the case of non-symmetric values. 

 \begin{corollary}
 \label{cor:stability}
 Assume that  battlefield  are balanced in the sense that there exits $r\in (0,1)$ such that
 $$
 \chi^2(v_A\|v_B) \vee \chi^2(v_B\|v_A) \le r^2
 $$
Then, the marginal distributions of the optimal Lotto strategy described in Section~\ref{SE:Lotto}  can be coupled into an optimal strategy for the corresponding Blotto game as long as
$$
\max_{i \in [n]} v_{B,i} \le \frac{T_B}{2T_A} (1-r)\,. 
$$

 \end{corollary}
 The proof of this result is based solely on computations; it is postponed to Section~\ref{SE:Proofcor:stability}

Note that the result of Corollary~\ref{cor:stability} is tight in the sense that if $r \to 0$ it recovers the result of Corollary~\ref{cor:symvals}. It is unclear whether the dependence in $r$ is sharp in our result and it is an interesting question to address in future work.

Under rather general conditions, the above two corollaries show the \emph{existence} of solutions with marginal distributions of the Lotto game derived in~\cite{KovRob21}. It remains to show that such a coupling may be realized efficiently. This is done in the next section.

\section{An efficient algorithm to compute solutions} \label{SE:Reductions}

Deriving solutions, either optimal strategies in the constant-sum setting or Nash equilibria, remains one of the major open problems surrounding the Blotto game. Previous attempts at this task have focused on deriving an \emph{explicit} coupling between marginals. This is possible in specific cases. For example, several explicit couplings between $n$ random variables $X_i \sim {\sf Unif}[0,1], i=1, \ldots, n$ are known~\cite{KnoSmi06,RusUck02}. In particular, this provides a solution to some Blotto problems with sufficient symmetry. However, this explicit approach fails for more general problems, and, in particular in the important case of asymmetric budget such as the one covered in Corollary~\ref{cor:symvals}. In this paper, we take another route by describing the efficient Algorithm~\LOTTOTOBLOTTO, whose pseudo-code is postponed to the Appendix \ref{SE:LOTTO2BLOTTO},  that computes an $\eps$-approximate solution with time complexity which is polynomial in $n$ and $1/\eps$.

In light of the previous section, our goal is to find an algorithm that efficiently computes a coupling between the marginal Lotto strategies described above. This task faces two major hurdles. 

On the one hand, the continuous nature of the marginals described above does not lend itself to efficient algorithms which typically work with discrete quantities. Instead, we propose to simply discretize the marginals at a scale of order $\eps>0$. In particular, this prevents us from replicating exactly the marginals of the Lotto game but we can show that the error employed in said discretization remains of the same order once propagated to the utility of a given player.

On the other hand, the mere description of a coupling between $n$ discrete marginals on $\cO(1/\eps)$ atoms is an object of size  $\cO(1/\eps^n)$, which is exponential in the number $n$ of battlefields. To overcome this limitation, we develop a careful scheme that allows us to reduce the problem to the case of $4$ marginals instead of $n$. 

Finally, we employ recent developments in computational optimal transport, to couple our 4 marginals using a variant of Sinkhorn iterations~\cite{Sin64, SinKno67, Cut13}.

\subsection{Reductions}

The typical size of a coupling with $n$ marginals is exponential in $n$. While this issue is, in general, hopeless to overcome, we can exploit some of the structure of the problem at hand. Indeed, a similar principle has been recently employed in multi-marginal optimal transport to devise polynomial-time algorithms under additional structure~\cite{AltBoi20}. More specifically, we reduce our problem to the case where there are only four marginals which remain mixable if the original marginals are mixable. 

This reduction is done in two steps. Recall that the marginals for the Lotto game described in Section~\ref{SE:Lotto} are either uniform distributions or mixtures of a uniform distribution with a Dirac point mass at zero. Our first step reduces to the case where $n-1$ marginals are uniform and only one is a mixture as above.  In our second step, we further reduce to the case where there are three uniform marginals and one mixture.

Throughout this section we focus on player $A$ for brevity. Reductions for player $B$ are analogous.

\subsubsection{Step 1: reduction to a single mixture}

The marginal distributions described in Section~\ref{SE:Lotto} consist of $|\cN(\gamma^*)|$ uniform distributions and $n-|\cN(\gamma^*)|$ mixtures of a uniform distribution and a Dirac point mass at 0 and our goal is to efficiently couple them into a joint mix coupling $\pi$ that has these marginals and satisfies the Blotto budget constraint. For clarity, we also regard uniform distributions as mixture distributions albeit with weight zero on the point mass. Otherwise, we say that a distribution is a \emph{strict} mixture. The goal of this first step is to reduce this coupling problem to the case where there are $n-1$ uniform distributions and one single strict mixture. To that end, we show that such a coupling $\pi$ may be obtained as a mixture of $n$ joint mixes $\pi_k, k=1, \ldots, n$:
$$
\pi=\sum_{k=1}^n q_k \pi_k\,, \quad q_k\ge 0,\ \sum_k q_k=1
$$
where the marginal distributions of $\pi_k$ consist of at most one strict mixture, the rest being uniform distributions. Moreover, this decomposition can be computed efficiently as the solution of a simple greedy procedure.

\begin{lemma}
\label{LEM:reduc_single_mixture}
Let $\gamma^*, \lambda^*$ be the parameter of a solution for the Lotto game and assume that the mixability condition~\eqref{EQ:blotto_cond_mix} 
 holds. Then, there exists a family $\pi_1, \ldots, \pi_n$ of couplings and a set of non negative weights $q_k\ge 0, \sum_k q_k =1$ such that
\renewcommand{\arraystretch}{2}
\begin{enumerate}

\item The marginal distributions of $(X_{A,1}^k, \ldots, X_{A,n}^k) \sim \pi_k$  are given by
$$
  X_{A,i}^k \sim  (1-p^{(k)}_i)\delta_0+ p^{(k)}_i
  {\sf Unif}\left[0, \frac{v_{A,i}\wedge(\gamma^* v_{B,i})}{\lambda^*} \right]   
$$
for some  $p^{(k)}_i \in [0,1], i, k \in[n]$ with at most one $p^{(k)}_i$ in $(0,1)$ for each $k$.
\item Each coupling $\pi_k, k=1, \ldots, n$ is a joint mix
\item The mixture of couplings 
\begin{equation}
    \label{EQ:mixcouplings}
\pi=\sum_{k=1}^n q_k \pi_k\,.
\end{equation}
is a solution for the Blotto game.
\item The total complexity of computing the  weights $q_k,p^{(k)}_i, i,k \in [n]$ scales as $\mathcal{O}\Big(n^2\log n\Big)$.

\end{enumerate}

\end{lemma}
Note that the mixture of couplings $\pi$ in~\eqref{EQ:mixcouplings}
is necessarily a joint mix as a mixture of joint mixes. To sample from it, Player $A$, simply samples $\pi_k$ with probability $q_k$ and plays according to the strategy prescribed by it.

The geometric proof of this lemma is delayed to Section \ref{SE:ProofLEM:reduc_single_mixture}, along with the pseudo-code of associated Algorithm \DECOMPOSE.

\begin{lemma}\label{Lemma:hypercube}
Fix $\ell \in \mathbb{R}_+^n$,  $T\in \mathbb{R}$ and let $\mathcal{H}=\{ x \in \mathbb{R}^n \ \text{s.t.}\ \langle \ell, x \rangle =T\}$ be an affine hyperplane. For any point $x \in \cC_n=\cH\cap[0,1]^n$, define  $\supp(x):=\{i \in [n], x_i \in (0,1)\}$; then there exist $\theta \in [0,1]$,  an extreme point $y\in \cC_n$ and another vector $\bar{y} \in \cC_n$ satisfying $\supp(\bar{y}) \varsubsetneq \supp(x)$ (in particular, $\bar{y}$ belongs to some $m$-face of $\cC_n$ where $m=|\supp(\bar{y})|<|\supp(x)|$) such that
$$
x= \theta y +(1-\theta)\bar{y}\ .
$$
The overall complexity  of computing $y,\bar{y}$ and $\theta$ is of order $\mathcal{O}(n\log n)$.
\end{lemma}
The proof of this Lemma is based on simple geometric arguments, and is postponed to Section \ref{SE:ProofLemma:hypercube} with the associated pseudo-code of Algorithm \Cube.

\subsubsection{Step 2. Reduction to four random variables}
The previous step reduces the joint mixability  problem of $n$ general mixtures, to a simpler one where at most one strict mixture is involved. Still, computing---in fact even describing---a coupling of $n$ variables requires generically exponential (in $n$) time and memory.  To overcome this limitation, we reduce the number of random variables from $n$ to a constant number. 

The following Lemma states that each $\pi_k$ can be realized as the coupling of 3 new uniform random variables and a strict mixture, thus reducing the mixability question from $n$ to only $4$ random variables. A careful inspection of the proof of Lemma~\ref{LM:reduc_to_four} below indicates that the reduction may lead to three marginals rather than four. In that case, two marginals are uniform and one is a strict mixture. To handle this case, some adjustments are needed; in particular---and obvisouly---with the size of the resulting coupling. However, extensions from four to three marginals are straightforward and we omit this case for clarity.

\begin{lemma}\label{LM:reduc_to_four}
Fix $n \ge 4$, $k \in [n]$, and assume without loss of generality that the last marginal of the coupling $\pi_k$ from Lemma~\ref{LEM:reduc_single_mixture} is a strict mixture. Then $(X_{A,1} , \ldots, X_{A,n})\sim \pi_k$ may be constructed from three uniform random variables $Y_1, Y_2, Y_3$ and a partition $\cI_0\sqcup\cI_1 \sqcup \cI_2 \sqcup \cI_3=[n-1]$ as follows. Set $X_{A,i} = 0$ for all $i \in \cI_0$, and 

    $$
    X_{A,i}=\theta_i Y_j\,, \qquad i \in \cI_j, \quad j \in \{1,2,3\}\,, 
    $$
    where $\theta_1, \ldots, \theta_{n-1} \in [0,1]$ are such that
    $$
    \sum_{i \in \cI_j} \theta_i =1, \quad j=1,2,3\,.
    $$

In particular, it holds that
$$
\sum_{i=1}^{n-1} X_{A,i}=Y_1+Y_2+Y_3 \quad \text{almost surely,}
$$
and $(Y_1, Y_2, Y_3, X_{A,n})$ are jointly mixable. The support of $Y_j$ is  $[0,b^*_j]$ where $$b^*_j=\sum_{i \in \cI_j} \frac{\gamma^* v_{B,i}}{\lambda^*} \wedge \frac{v_{A,i}}{\lambda^*}$$
Moreover, the $\theta_i$'s, the sets $\cI_j$, and the parameters of the distributions of $Y_1, Y_2, Y_3$ can each be computed in constant time.

\end{lemma}
The  proof of this Lemma, based on standard mixability arguments, is postponed to Section~\ref{SE:ProofLM:reduc_to_four}, with the pseudo-code of the corresponding Algorithm \REDUCTION.

Note that any joint mix of $(Y_1, Y_2,Y_3,X_{A,n})$ readily yields a joint mix of $(X_{A,1}, \ldots, X_{A,n})$ by defining $X_{A,i}=\theta_i Y_{j(i)}$, where $j(i) \in [3]$ is the unique integer such that $i \in \cI_{j(i)}$.

\subsection{Discretization}
\label{SE:Step3}

The problem of finding a solution for the Blotto game has been reduced to the construction of a coupling of (at most) four random variables,  three of them being uniform over some intervals  and the fourth one being a mixture between a Dirac mass at zero and some uniform distribution. Throughout this section we denote these random variables as $Y_1, \ldots, Y_4$ for simplicity; in the notation of the previous section, they correspond to $Y_1, Y_2, Y_3$, and $X_{A,n}$ respectively.

Unfortunately, even in this simple case, finding explicit, closed-form, couplings appears to be possible only under stringent additional conditions that limit the scope of the Blotto game.
To overcome this limitation, we take an algorithmic approach, describing an efficient way to find an approximate solution. To that end, we obviously need to work with discrete random variables and describe here a coupling between these discretized random variables.

Let $(Y_1, Y_2,Y_3, Y_4) \sim \varpi$ be jointly mixed so that
\begin{equation}
    \label{EQ:JM4}
    Y_1+Y_2+Y_3+Y_4=T_A\,.
\end{equation}
Moreover, let $h>0$ be some (small) discretization parameter. Define the quantized random variables $\tilde Y_i$ by 
\begin{equation}
    \label{EQ:margin_Const}
    \tilde Y_i =\left\lfloor \frac{Y_i}{h}\right\rfloor, \  i=1,2, 3\,, \qquad 
\tilde Y_4=\frac{T_A}{h}-\left\lfloor \frac{T_A-Y_4}{h} \right\rfloor \,.
\end{equation}

Our goal is to compute any of the joint distributions $\mathsf{D}(\varpi)$ of the vector $\tilde Y=(\tilde Y_1, \tilde Y_2, \tilde Y_3, \tilde Y_4)$ when $\varpi$ ranges over joint mixes. 

As a first step towards this goal, note that these discretized random variables need not be jointly mixable. Indeed, in general we have $\tilde Y_1 + \tilde Y_2 + \tilde Y_3 \le \lfloor (Y_1+Y_2+Y_3)/h\rfloor$ but equality may fail to hold because of discretization errors. To account for these, let $\eps \in \{0,1,2\}$ be  defined as
\begin{equation}
    \label{EQ:def_epsilon}
    \eps=(T_A/h-\tilde Y_4) -(\tilde Y_1 + \tilde Y_2 + \tilde Y_3)\,,
\end{equation}
and consider the augmented random vector $\tilde Y_{+}= (\tilde X, \eps)$. In light of~\eqref{EQ:def_epsilon}, $\tilde Y_{+} \in \R^5$ lives almost surely on a four dimensional subspace. As such, its distribution may be represented by a 4-tensor $(\Gamma_{ijke})$ with entries given by 
$$
\Gamma_{ijke}=\p(\tilde Y_1=i, \tilde Y_2=j, \tilde Y_3=k, \eps=e)\,.
$$
In particular, note that $e \in \{0,1,2\}$ while $i,j,k$ each range in a set of integers of size $\Theta(1/h)$. Using~\eqref{EQ:def_epsilon} we can read off the distribution of $\tilde Y_4$ from this tensor. 

This tensor is subject to four sets of linear constraints, one for each of the marginal constraints given in~\eqref{EQ:margin_Const}. They are given by
$$
\sum_{jke}\Gamma_{ijke}=\p(\tilde Y_1 =i) \ \forall i\,, \quad \sum_{ike}\Gamma_{ijke}=\p(\tilde Y_2 =j)\  \forall j\,, \quad \sum_{ije}\Gamma_{ijke}=\p(\tilde Y_3 =k)\ \forall k\,,
$$
and, in light of~\eqref{EQ:def_epsilon}, by
$$
\sum_{i+j+k+e=\ell}\Gamma_{ijke}=\p(T_A/h-\tilde Y_4 =\ell)\ \forall \ell\,.
$$

Note that indeed, any draw from a distribution that satisfies the above constraints yields a random vector $(\tilde Y_1, \tilde Y_2, \tilde Y_3, \eps)$. Defining $\tilde Y_4$ by solving~\eqref{EQ:def_epsilon} yields a vector $\tilde Y  \sim {\sf D}(\varpi)$ for some joint mix $\varpi$ defined as above. In other words, $\tilde Y$ is indeed the discretization of random variables drawn from a joint mix (though it need not be jointly mixable itself).

\bigskip

Since the random variables $Y_j$, for $j \in [3]$, constructed at the previous step have a support equal to $[0,b^*_j]$ where $b^*_j=\sum_{i \in \cI_j} \frac{\gamma^* v_{B,i}}{\lambda^*} \wedge \frac{v_{A,i}}{\lambda^*}$, the reduced (to 4 random variables) and discretized problem reduces to  finding some tensor $(\Gamma_{ijke})$ with $ 3\cdot\left\lfloor\frac{b^*_1}{h}+1\right\rfloor\cdot\left\lfloor \frac{b^*_2}{h}+1\right\rfloor\cdot\left\lfloor\frac{b^*_3}{h}+1\right\rfloor$ entries satisfying at most $4\cdot \left\lfloor\frac{\max b^*_j}{h}+1\right\rfloor$ linear constraints. Although this can be done simply via linear programming (hence polynomially in $h^{-1}$, more precisely in $\tilde{\mathcal{O}}(1/h^{8,5})$ with Vaidya's algorithm), a quite  efficient and more popular way is to use a variant of Sinkhorn-Knopp algorithm that quickly finds approximated solutions. This is more relevant as this linear program is already some approximation of the original problem, hence  there is no point of solving  it exactly.

The pseudo-code of the Algorithm \DISCRETIZE\ can be found in Section \ref{SE:DISCRETIZE}

\subsection{Tensor scaling using Sinkhorn iterations}
\label{SE:Step4}

In light of the previous sections, we have reduced our problem to that of finding coupling in the form of a 4-tensor $(\Gamma_{ijke}, i \in [d_1], j\in [d_2], k \in [d_3], e\in \{0,1,2\})$ with non-negative entries subject to marginal constraints. We approach this problem from a computational perspective and propose and algorithm that converges rapidly to a feasible solution. To describe this algorithm, recall that its input are four probability vectors $\tilde \mu_j \in \R^{d_j}, j=1,\dots, 4$, with $d_4=d_1+d_2+d_3+2$ that represent the probability mass functions of the discretized random variable $\tilde  Y_1,  \tilde Y_2, \tilde Y_3, T_A/h-\tilde Y_4$ defined in the previous section: $\tilde \mu_{j,\boldsymbol{\cdot}}=\p(\tilde Y_j=\boldsymbol{\cdot})$, $j=1, \ldots, 3$, $\tilde \mu_{4,\boldsymbol{\cdot}}=\p(T_A/h-\tilde Y_j=\boldsymbol{\cdot})$.

The linear constraints take the form
\begin{align}
   \bar \Gamma_i^{(1)}&:= \sum_{jke}\Gamma_{ijke}=\tilde \mu_{1,i} \ &\forall i \in [d_1] \,,\\
   \bar \Gamma_j^{(2)}&:=   \sum_{ike}\Gamma_{ijke}=\tilde \mu_{2,j}\  &\forall j\in [d_2]\,, \\  \bar \Gamma_k^{(3)}&:= \sum_{ije}\Gamma_{ijke}=\tilde \mu_{3,k}\ &\forall k\in [d_3]\,,\\
    \bar \Gamma_\ell^{(4)}&:=  \sum_{i+j+k+e=\ell}\Gamma_{ijke}=\tilde \mu_{4,\ell} \ &\forall \ell\in [d_4]\,.
\end{align}
Denote by $\cG$ the set of tensors $\Gamma=(\Gamma_{ijke})$ that satisfy these constraints.

To solve this problem, we propose to project the all-ones tensor ${\bf 1}$ onto $\cG$ using the Kullback-Leibler (KL) divergence. Recall that the KL divergence between two nonnegative tensors $\Gamma, \Gamma'$ is given by
$$
\KL(\Gamma\|\Gamma')=\sum_{ijke} \Gamma_{ijke} \log \left(\frac{ \Gamma_{ijke} }{ \Gamma'_{ijke} }\right)
$$
In particular, ${\sf KL}(\Gamma\|{\bf 1})$ is simply the (negative) entropy $\sH(\Gamma)$ of $\Gamma$ and we aim to solve the convex optimization problem
$$
\min_{\Gamma \in \cG} \sH(\Gamma)=\sum_{ijke} \Gamma_{ijke} \log \left(\Gamma_{ijke}\right)\,.
$$
While many algorithms are available to solve this problem~\cite{bubeckConvexOptimizationAlgorithms2015}, its specific structure can be exploited efficiently. Indeed, first order optimality  conditions imply that any optimal $\Gamma$ must be of the form
\begin{equation}
    \label{EQ:couplingform}
    \Gamma_{ijke}=\xi_{1,i}\cdot \xi_{2,j}\cdot\xi_{3,k}\cdot\xi_{4,i+j+k+e}\,,
\end{equation}
for some scaling vectors $\xi_j \in (0, \infty)^{d_j}, j=1,\dots, 4$. This representation readily calls for an iterative tensor scaling algorithm similar to the Sinkhorn algorithm~\cite{Sin64, SinKno67, Cut13}. Tensor scaling has been investigated in more classical setups~\cite{lin2020complexity, AltBoi20} that slightly differ from the present setup because the fourth marginal constraint takes a special form. Nevertheless, the implementation of Algorithm \SINKHORN\ remains straightforward and is presented in Section~\ref{SE:SINKHORN}. Its analysis is also a straightforward  extension of that for the traditional matrix case~\cite{AltWeeRig17}. More specifically, following the exact same lines as the one of  Theorem 4.3 in \cite{lin2020complexity}, we readily get the following result.

\begin{proposition}\label{PR:Sinkhorn_complexity}
Define 
$$
\tilde \mu_{\min}=\min_{i\in[4],j, \tilde \mu_{i,j} \neq 0} \tilde \mu_{i,j}\,.
$$
 Algorithm \SINKHORN\ terminates and returns a tensor ${\Gamma}$ such that $\sum_{i=1}^4\|\bar \Gamma^{(i)}-\tilde \mu_i\|_1 \leq \eta$ after at most $32\eta^{-1}(1-\log \tilde \mu_{\min})$ iterations. Moreover, each marginal $\bar \Gamma^{(i)}$ of $\Gamma$ has positive entries that sum to one and hence is a probability vector.

\end{proposition}

What have we accomplished so far? Through several reductions and a tensor scaling algorithm, given the datum of a Blotto game, we are able to compute a joint distribution that corresponds to an approximate solution. In Section~\ref{SE:Errors}, we evaluate the accuracy of this approximation in terms of the value of the game by showing that the various approximations (discretization and numerical precision of the algorithm) do not blow up when propagated back into the reductions. Before that, we investigate an important operational question: how to sample a strategy from the resulting coupling $\Gamma$. 

\subsection{From coupling to sampling}

\label{sec:smoothing}
Finding an efficient construction of (approximate) equilibria or optimal strategies is only relevant if it can be  associated to some efficient sampling method so that  a player may query a sampler and receive the allocation $(X_{A,1},\ldots,X_{A,n})$  that they should play on each battlefield.  In light of the various reduction steps employed above, it is sufficient to sample a 4-tuple 
$$
(\tilde Y_1, \tilde Y_2, \tilde Y_3,\varepsilon) \in [0,b^*_1/h]\times[0,b^*_2/h]\times [0,b^*_3/h]\times \{0,1,2\}
$$
from the output $\Gamma$ of Algorithm \SINKHORN.
Indeed, from $(\tilde Y_1, \tilde Y_2, \tilde Y_3,\varepsilon)$, we obtain the random variables $Y_i, i=1,\dots, 4$ that are approximately distributed from the joint mix $\varpi$ as follows.

To ensure that the marginal distributions are continuous, let $U \sim {\sf Unif}[0,1]$ and define
\begin{align*}
Y_1' =  (\tilde Y_1+\frac{\varepsilon}{3}+\frac{U}{3})h \wedge  b^*_1, \quad 
Y_2' = (\tilde Y_2+\frac{\varepsilon}{3}+\frac{U}{3})h\wedge b^*_2 , \quad
Y_3' = (\tilde Y_3+\frac{\varepsilon}{3}+\frac{U}{3})h \wedge  b^*_3.
\end{align*}
To correct for potential boundary effects, define $S=Y_1'+Y_2'+Y_3'$ and
$$
\zeta = \mathds{1}\{S>T_A\} \frac{T_A-S}{3}+\mathds{1}\{S<T_A-b^*_4\} \frac{T_A-b^*_4-S}{3}.
$$
Then take $\bar Y_1=Y_1'+\zeta$, $\bar Y_2=Y_2'+\zeta$, $\bar Y_3=Y_3'+\zeta$, and $\bar Y_4=T_A-(Y_1+Y_2+Y_3)$.

We call this procedure the \emph{smoothing} procedure. Finally, as mentioned before, just define $X_{A,i}=\theta_i \bar Y_{j(i)}$, where $j(i) \in [3]$ is the unique integer such that $i \in \cI_{j(i)}$.

Note that the random variable $U\sim {\sf Unif}[0,1]$ is superfluous and theoretical results would follow by taking $U=0$. Its role is simply to ensure, for cosmetic reasons, that the random marginal distributions are continuous apart from the potential point mass at zero.
\medskip

It remains to sample $(\tilde Y_1, \tilde Y_2, \tilde Y_3,\varepsilon)$ from the output $\Gamma$ of Algorithm~\SINKHORN. This is quite straightforward in light of the factored form of $\Gamma$. Indeed, recall that the coupling output by Algorithm~\SINKHORN\ has the form~\eqref{EQ:couplingform}.
$$
\Gamma_{ijke}=\xi_{1,i}\cdot \xi_{2,j}\cdot\xi_{3,k}\cdot\xi_{4,i+j+k+e}\,, \ \forall  i \in [d_1], j\in d_2, k \in [d_3], e\in \{0,1,2\}
$$
As a consequence, we can draw from $\Gamma$ as follows: 
\begin{enumerate}
    \item Set $\tilde Y_1=i  \in [d_1]$ with probability proportional to $\xi_{1,i}$
    \item  Set $\tilde Y_2=j \in [d_2]$ with probability proportional to $\xi_{2,j}$
    \item  Set $\tilde Y_3=k \in [d_3]$ with probability proportional to $\xi_{3,k}$
    \item Conditionally on $(\tilde Y_1, \tilde Y_2, \tilde Y_3)$, set $\varepsilon=e \in \{0,1,2\}$ with  probability proportional to $\xi_{4,\tilde Y_1+ \tilde Y_2+ \tilde Y_3+e.}$
\end{enumerate}
The pseudo-code of Algorithm \SAMPLE\ can be found in Section~\ref{SE:SAMPLE}.

\section{Approximation errors and computational complexity}
\label{SE:Errors}

The construction of the previous section relies on various approximations, each of them inducing some error that can be mitigated at the cost of additional computational complexity by tuning the discretization parameter $h$ of Section~\ref{SE:Step3} and the tolerance parameter $\eta$ in Algorithm \SINKHORN. In this section we study the computational complexity required to reach an $\eps$-approximate solution.

\subsection{From approximate strategies to approximate solutions}

Note that the very notion of   ``approximate solution'' strongly depends on whether the problem is value-symmetric  ($v_{A,i}=v_{B,i}$ for all $i$) or -asymmetric ($v_{A,i} \neq v_{B,i}$ for some $i$). Indeed, in the former case, the game is constant-sum and  optimal strategies do exist. This is no longer true in the latter case where only Nash equilibria are considered. As a consequence, we can consider approximation of a single optimal strategies in value-symmetric games, while we will have to consider approximations of a pair of equilibrium strategies in value-asymmetric ones. In the following two sections, we consider each case separately. In the remaining, we shall focus the analysis  on Player $A$, but it is almost identical for player $B$; hence we do not repeat it for the sake of clarity.

\subsubsection{The value-symmetric case}

A value-symmetric Blotto game, where $v_{A,i}=v_{B,i}=v_i$ for all $i$ is constant-sum and optimal strategies exist for each player. In particular, this allows us to provide strong approximation guarantees by controlling how sub-optimal the expected utility of a player is.

To check this well-known fact on our specific instance, consider the utility of player $A$. 
Set two equilibrium parameters $\gamma^*=T_B/T_A\le 1$ and $\lambda^*=T_B/(2T_A^2)$ (see Corollary~\ref{cor:symvals})  defining an optimal strategy and observe that $\cN(\gamma^*)=[n]$ since $\gamma^*\le 1$. For $i \in [n]$, let $X_{A,i}\sim {\sf Unif}[0, 2T_A v_i]$ denote the amount allocated by player $A$ to battlefield $i$ according to this optimal strategy and denote  by $F_{A,i}$ its cumulative distribution function (cdf). The expected utility (a.k.a. reward)  of player $A$ if player $B$ chooses  allocation  $x_B=(x_{B,i})_i$ depends only on the sequence  $F_{A}=(F_{A,1},\dots, F_{A,n})$ of marginal cdfs rather than the whole coupling. It is given by
\begin{align*}
\cU_A(F_A, x_B)=\sum_{i=1}^n v_{A,i}\p(X_{A,i}> x_{B,i})
=\sum_{i=1}^n v_i \big(1-F_{A,i}(x_{B,i})\big)
=1 - \sum_{i=1}^n v_i \wedge \frac{x_{B,i}}{2T_A}\ge 1 - \frac{T_B}{2T_A}\,.
\end{align*}
where we used the fact that the $v_i$'s sum to 1 and the $x_{B,i}$'s sum to at most $T_B$. Moreover, if $B$ employs the mixed strategy $X_{B,1}, \ldots, X_{B,n}$ described in Corollary~\ref{cor:symvals}, the utility of player $A$, denoted $\cU_A(F_A, F_B)$, changes as follows. Let $ U_{A,i}, U_{B,i}\sim {\sf Unif}[0, 2T_A v_i]$ be a sequence of uniform random variables such that $ U_{A,i}$ is independent of $U_{B,i}$. In particular, $\p(U_{A,i}>U_{B,i})=.5$ and 
\begin{align*}
\cU_A(F_A, F_B)
&=\big(1-\frac{T_B}{T_A}\big)\sum_{i=1}^n v_i \big(1-F_{A,i}(0)\big)+\frac{T_B}{T_A}\sum_{i=1}^n v_i \p(U_{A,i}>U_{B,i})\\
&=\big(1-\frac{T_B}{T_A}\big)+\frac{T_B}{2T_A}= 1 - \frac{T_B}{2T_A}\,.
\end{align*}
In particular, the strategy of player $A$ given in Corollary~\ref{cor:symvals} is optimal and its optimal utility is given by $1-T_B/(2T_A)$.

To estimate the cost of the various approximations incurred by player $A$, let $X_{A,i}^{h,\eta}\sim P_{A,i}^{h,\eta}$ denote the strategy on battlefield $i \in [n]$. The notation is meant to emphasize that the approximation error stems from two sources: the precision level $\eta$ of Algorithm \SINKHORN\ and the grid size $h$ of discretization  procedure in Section~\ref{SE:Step3}. In particular, we write $P_{A,i}^{0,0}:=P_{A,i}$. Cognizant  of this approximation error, player $A$ may take advantage of the suboptimality of the strategy of player $A$ and respond with best-response strategy denoted $X_{B}^{h,\eta}=(X_{B,1}^{h,\eta}, \ldots, X_{B,n}^{h,\eta})$. As a result, the suboptimality gap of player $A$'s expected utility is controlled as follows
\begin{align*}
\cU_A(F_A, F_B)-\cU_A(F_A^{h,\eta}, F_B^{h,\eta})&=  \sum_{i=1}^n v_{i}\p(X_{A,i}> X_{B,i})-  \sum_{i=1}^n v_{i}\p(X_{A,i}^{h,\eta}> X_{B,i}^{h,\eta})\\
    &\le   \sum_{i=1}^n v_{i}\p(X_{A,i}> X_{B,i}^{h,\eta})-  \sum_{i=1}^n v_{i}\p(X_{A,i}^{h,\eta}> X_{B,i}^{h,\eta}) \\
    &\le   \sum_{i=1}^n v_{i}\sup_{x> 0} \left[\p(X_{A,i}> x)- \p(X_{A,i}^{h,\eta}> x)\right]
\end{align*}
where in the first inequality, we use the fact that $ X_{B}$ is an optimal response for $B$ when $A$ plays $X_{A}$. 

It will be convenient in the sequel to further bound the above quantity using the   $\infty$-Wasserstein distance---see~\cite[Section~5.5.1]{San15}---between $P_{A,i}$ and $P_{A,i}^{h,\eta}$, denoted $W_{\infty}(P_{A,i},P_{A,i}^{h,\eta})$. Indeed, to show that
$W_{\infty}(P_{A,i},P_{A,i}^{h,\eta})\le \omega$ for some $\omega\ge 0$, it is sufficient to exhibit a coupling of $X_{A,i}, X_{A,i}^{h,\eta}$ such that $|X_{A,i}- X_{A,i}^{h,\eta}|< \omega$ almost surely. Below, we often do so implicitly as such couplings are, in all instances, trivial. 

Fix $\omega >0$ and assume $W_{\infty}(P_{A,i},P_{A,i}^{h,\eta})\le \omega$. Then for any $\omega\ge0$,  we have
$$
\p(X_{A,i}> x)- \p(X_{A,i}^{h,\eta}> x) \le \p(x< X_{A,i} <x+\omega ) \le \frac{\lambda^*\omega}{\gamma^*v_i}\,,
$$
where we used the fact that $X_{A,i} \sim {\sf Unif}[0,\gamma^*v_i/\lambda^*]$. The above two displays together yield that the suboptimality gap for player $A$ is controlled as
$$
\cU_A(F_A, F_B)-\cU_A(F_A^{h,\eta}, F_B^{h,\eta}) \le \frac{\lambda^*}{\gamma^*} \sum_{i =1}^n W_{\infty} (P_{A,i},P_{A,i}^{h,\eta})\,.
$$

Recall that  Step 1 in the reduction consists in decomposing $P_{A,i}$  as a mixture of (at most) $n$ other distributions, i.e., $P_{A,i}=\sum_k q_kP_{A,i}^{(k)}$. Accordingly, we also have constructed $P_{A,i}^{h,\eta}$ as a mixture $P_{A,i}^{h,\eta}= \sum_k q_kP_{A,i}^{h,\eta,(k)}$. It follows readily from the definition of $W_\infty$ that that
$$ W_{\infty}(P_{A,i},P_{A,i}^{h,\eta})\leq \sum_k q_k W_{\infty}(P_{A,i}^{(k)},P_{A,i}^{h,\eta,(k)}).$$
In particular controlling each term on the right-hand side uniformly in $k$ results in the same control on the desired error. Therefore, without loss of generality, we may assume that $P^{(k)}_{A,i}=P_{A,i}$ and $P^{h,\eta, (k)}_{A,i}=P^{h,\eta}_{A,i}$ so as to keep the notation light. Moreover, as above, we assume without loss of generality the last marginal $P_{A,n}$ is the only strict mixture.

\subsubsection{The value-asymmetric case}
When values are asymmetric, the game is no longer constant-sum and we shift our focus from optimal strategies to Nash equilibria. In this context the notion of approximation is more subtle and has to be carried out jointly for both players.

For any set of marginal cdfs $G_A=(G_{A,1}, \ldots, G_{A,n})$ and $G_B=(G_{B,1}, \ldots, G_{B,n})$, denote by $\cU_A(G_A,G_B)$ the expected utility of player $A$ if player $A$ plays according to strategy $G_A$ while player $B$ plays according to strategy $B$.

A pair $(F_A,F_B)$ is a Nash equilibrium if
$$
\cU_A(F_A,F_B) \ge \cU_A(G_A,F_B), \text{ for all admissible} \  G_A\,.
$$
Writing
\begin{align*}
    \cU_A(F_A,F_B)&=\cU_A(F_A,F_B)-\cU_A(F_A^{h,\eta},F_B)+\cU_A(F_A^{h,\eta},F_B)-\cU_A(F_A^{h,\eta},F_B^{h,\eta})+\cU_A(F_A^{h,\eta},F_B^{h,\eta})\,,\\
     \cU_A(G_A,F_B)&= \cU_A(G_A,F_B)- \cU_A(G_A,F_B^{h,\eta})+ \cU_A(G_A,F_B^{h,\eta})
\end{align*}
it readily follows from the above two displays that
$$
\cU_A(F_A^{h,\eta},F_B^{h,\eta}) +\eps \ge \cU_A(G_A,F_B^{h,\eta}), \text{ for all admissible} \  G_A\,,
$$
where, using similar computations as above, 
$$
\eps \le \sum_{i=1}^n v_{A,i}\left[\frac{\lambda^*}{\gamma^* v_{B,i}}W_\infty(P_{A,i}, P_{A,i}^{h,\eta})+\frac{2\lambda^*}{ v_{A,i}} W_\infty(P_{B,i}, P_{B,i}^{h,\eta}) \right]
    $$

\subsection{Control of the errors}

In both cases, symmetric or asymmetric values, a control of the approximation error follows from controlling    $W_{\infty}(P_{A,i},P_{A,i}^{h,\eta})$. In the rest of this section, we sightly abuse notation and write $W_\infty(X,Y)$ when the distributions of the random variables $X$ and $Y$ are clear from the context.

Recall that for any $i \in [n]$ since $X_{A,i} =\theta_i Y_{k(i)}$ and similarly for the approximate versions, for some fixed $\theta_i \in [0,1]$, we have for any $h, \eta \ge 0$ that
$$
 W_{\infty}(P_{A,i},P_{A,i}^{h,\eta}) = \theta_i W_{\infty}(Y_{j(i)},\bar Y_{j(i)})
$$
where $Y_j$ is the result of the reductions and is defined in Lemma~\ref{LM:reduc_to_four} while $\bar Y_j$ is the output the smoothing procedure and is defined in Section~\ref{sec:smoothing}.

Recall that the discrepancy between $\bar Y_j$ and the target $Y_j$ stems from three approximations: discretization error ($h>0$), numerical error ($\eta>0$), and the error due to smoothing step. 
The error coming from the smoothing step is easy to control: for any $j \in [4]$, we have $W_{\infty}(Y_j,\bar  Y_j) \le W_{\infty}(Y_j,h\tilde Y_j) + h$. We have proved that
\begin{equation}
    \label{EQ:smootherr}
     W_{\infty}(P_{A,i},P_{A,i}^{h,\eta}) \le \theta_i W_{\infty}(Y_{j(i)},h\tilde  Y_{j(i)})+ \theta_i h
\end{equation}

As a result, it is sufficient to control the discretization error and the numerical error at the level of the variables $\tilde Y_j$. To emphasize the presence of these errors, we employ the same notation as for $P_{A,i}^{h, \eta}$ and write $\tilde Y_j=\tilde Y_{j}^{h, \eta}$ for $h>0$, $\eta \ge 0$. By the triangle inequality, we have
$$
W_{\infty} (Y_j,h \tilde Y_j) \le \underbrace{W_{\infty}(Y_j,h\tilde Y_j^{h,0}) }_{\text{discretization error}} +  \underbrace{W_{\infty}(h\tilde Y_j^{h,0}, h\tilde Y_j^{h,\eta}) }_{\text{numerical error}}\,.
$$

The discretization error is trivial to control. Indeed, in light of the coupling provided by~\eqref{EQ:margin_Const},  we get that 
\begin{equation}
    \label{eq:discrerr}
    W_{\infty}(Y_j,h\tilde Y_j^{h,0})\le h
\end{equation}

Finally, to control the numerical error $W_{\infty}(h\tilde Y_j^{h,0}, h\tilde Y_j^{h,\eta})$, recall that the tolerance $\eta>0$ in Algorithm \SINKHORN\ controls the $\ell_1$ error between the current marginals and the targets. Hence, we need to bound the $\infty$-Wasserstein distance by the $\ell_1$ distance. This is quite straightforward since $\tilde Y_j^{h,\eta}$ has bounded support for $\eta\ge 0$. Indeed, recall from Lemma~\ref{LM:reduc_to_four} that for any $j \in [4]$ we have that $\tilde Y_j^{h,\eta} \in [0,b_j^*/h]$, where 
$$
b_j^*=\frac{1}{\lambda^*}\sum_{i \in \cI_j}(\gamma^* v_{B,i})\wedge v_{A,i}\,.
$$
Hence, $W_{\infty}(h\tilde Y_j^{h,0}, h\tilde Y_j^{h,\eta}) \le b_j^* \eta$, and we have proved that
$$
 W_{\infty}(P_{A,i},P_{A,i}^{h,\eta}) \le  2\theta_ih + \theta_ib_{k(i)}^*\eta\,.
$$

In particular, for the value-symmetric case, since 
$$
b_{k(i)}=\frac{\gamma^*}{\lambda^*} \sum_{l \in \cI_{k(i)}}v_{l}\le \frac{\gamma^*}{\lambda^*}  \,, 
$$
we get the following simple bound
$$
 W_{\infty}(P_{A,i},P_{A,i}^{h,\eta})\le 2\theta_i h + \frac{\theta_i\gamma^*\eta}{\lambda^*}\,.
$$
Since $\sum_i \theta_i=4$, the suboptimality gap for player $A$ is controlled as
\begin{equation}
    \label{EQ:final_bound_sym}
    \cU_A(F_A)-\cU_A(F_A^{h,\eta}) \le \frac{\lambda^*}{\gamma^*} \sum_{i =1}^n W_{\infty} (P_{A,i},P_{A,i}^{h,\eta}) \le 8\frac{\lambda^*}{\gamma^*}h + 4 \eta\,.
\end{equation}

In the value-asymmetric case, we get a suboptimality for player $A$ smaller than 
\begin{equation}
    \label{EQ:final_bound_nash}
\eps_A = \left(16+\frac{8}{\gamma^*}\max_i \frac{v_{A,i}}{v_{B,i}}\right) \lambda^* h + \left(8\gamma^*+ 4\max_i \frac{v_{A,i}}{v_{B,i}} \right) \eta
\end{equation}
and, with symmetric arguments, a suboptimality for player $B$ smaller than 
$$
\eps_B = \left(16+8\max_i \frac{v_{B,i}}{v_{A,i}}\right) \frac{\lambda^* h}{\gamma^*} + \left(8+4\max_i \frac{v_{B,i}}{v_{A,i}} \right) \eta\, .
$$

\subsection{Computational complexity}
In this section, we tally the complexity required to achieve either $\eps$-suboptimality gap in the value-symmetric case or an $\eps$-Nash equilibrium in the value-asymmetric case. 

Before making this distinction recall the various steps that were employed, together with their computational complexity.

\begin{description}
    \item[Lotto Step.] Computing one (and, actually, all) pair of parameters $(\gamma,\lambda)$ requires $\mathcal{O}(n\log n)$ operations, see Proposition~\ref{PR:Gamma}.
    \item[Step 1:] Computing all couplings $\pi_k$ and their associated convex weights $q_k$ requires $\mathcal{O}(n^2)$ operations; see Lemma~\ref{LEM:reduc_single_mixture}.
    \item[Step 2:] Given some coupling $\pi_k$ computed at step 1, the reduction from $n$ to only $4$ random variables requires $\mathcal{O}(n^2)$ operations; see Lemma~\ref{LM:reduc_to_four}.
    \item[Step 3:] The discretization step is computationally costless.
    \item[Step 4:] The Sinkhorn algorithm requires $\mathcal{O}(\frac{\log(1/\tilde{\mu}_{\min})}{\eta})$ operations; see Proposition~\ref{PR:Sinkhorn_complexity}. Since the marginal distributions $\mu_{i}$ are $h$-discretizations of either uniform on an interval of size at most $\sum \frac{\gamma^* v_{B,i}}{\lambda^*}\leq \frac{\gamma^*}{\lambda^*}$, or uniform on interval of length $\frac{v_{A,i}}{\lambda^*}$ with weight $\frac{v_{A,i}}{\gamma^* v_{B,i}}$, it holds that $\tilde{\mu}_{\min} \geq \frac{h\lambda^*}{\gamma^*}$. At each iteration of Sinkhorn, all the components of the tensor are computed, hence a complexity, per iteration, in  $\mathcal{O}\Big((\frac{\gamma^*}{h\lambda^*})^3\Big)$.
    \item[Step 5:] The sampling cost comes from the reconstruction of the budget allocation from the 4 random variables, constructed at Step 2 and sampled from the coupling computed at step 4. The sampling step has a linear cost with respect to the discretization size $\mathcal{O}(\frac{\gamma^*}{h\lambda^*})$ while the reconstruction complexity scales linearly with respect to the number of battlefields  $\mathcal{O}(n)$ .  
\end{description}

We are now in a position to state our main theorems. We begin with the symmetric-value case for player $A$. The result for player $B$ is completely analogous and therefore omitted. 

\begin{theorem}
\label{TH:complex_cstsum}
Consider the two-player Blotto game on $n$ battlefields with symmetric values $v_1,\ldots, v_n$ where player $A$ has budget $T_A$ and player $B$ has budget $T_B \le T_A$  and assume that these data satisfy the conditions of Corollary~\ref{cor:symvals}. 
Fix $\eps>0$ and let $\eta=\eps/4$ and $h=\eps T_A/8$. Then the procedure described in Algorithm \SAMPLE\ samples from an $\eps$-suboptimal strategy for player $A$ in time 
$$
\cO\left( n^2 + \frac{\log (1/\eps)}{\eps^4}\right)
$$
\end{theorem}
\begin{proof}
Note first  that the preprocessing cost associated to steps 1 through 4 is $\cO(n^2)$. 

To compute the cost of Sinkhorn iterations, observe that the parameters $\eta$ and $h$ are chosen in such a way that each term on the right-hand side of~\eqref{EQ:final_bound_sym} is equal to $\eps/2$:
$$
8\frac{\lambda^*}{\gamma^*}h = 4 \eta=\frac{\eps}{2}
$$
Hence,
$$
\tilde \mu_{\min} \ge \frac{h \lambda^*}{\gamma^*}=\frac{\eps}{16}\,,\quad \text{and} \quad 
\left(\frac{\gamma^*}{h \lambda^*}\right)^3=\left(\frac{16}{\eps}\right)^3\,.
$$
Moreover, since $\eta=\eps/8$, we get that the total complexity of Sinkhorn iterations is
$$
\cO\left(\left(\frac{\gamma^*}{h \lambda^*}\right)^3\frac{\log\left(1/\tilde \mu_{\min}\right)}{\eta}\right)=\cO\left(\eps^{-3} \frac{\log\left(1/\eps\right)}{\eps}\right)=\cO\left(\eps^{-4} \log(1/\eps)\right)
$$
Finally, the last step has a total cost of $\cO(n+\eps)$ which is negligible with respect to the combination of previous steps.
\end{proof}
It is worth noting that in the value-symmetric case, the computational complexity of our procedure is \emph{independent} of the datum of the problem (budgets and values) under  the normalization~\eqref{EQ:sumvi}. Note that this normalization merely scales the utility and should of course affect the desired accuracy parameter~$\eps$.

\medskip

We now move to the asymmetric-value case and characterize the complexity of our procedure to compute an $\eps$-Nash equilibrium for the Blotto game. As above, we focus on player $A$ only.

\begin{theorem}
\label{TH:complex_nash}
Consider the   Blotto game on $n$ battlefields with asymmetric values $v_{P,1},\ldots, v_{P,n}$, $P\in \{A,B\}$, where player $P \in \{A,B\}$ has budget $T_P$, with $T_B \le T_A$  and assume that these data satisfy the conditions of Corollary~\ref{cor:symvals}. Define
$$
\mathfrak{m}=\max_i \frac{v_{A,i}}{v_{B,i}} \vee \frac{v_{B,i}}{v_{A,i}}\,.
$$
Fix $\eps>0$ and let 
$$
\eta=\frac{\eps}{24\mathfrak{m}} \,, \qquad h= \frac{\gamma^*}{\lambda^*}\frac{\eps}{48 \mathfrak{m}} \,.
$$
Then Algorithm~\LOTTOTOBLOTTO\ samples from an $\eps$-Nash equilibrium in time
$$
\cO\left(n^2+\left(\frac{\mathfrak{m}}{\eps}\right)^{4} \log\left(\frac{ \mathfrak{m}}{\eps}\right)\right)
$$
\end{theorem}
\begin{proof}
Note first  that the preprocessing cost associated to steps 1 through 4 is $\cO(n^2)$. 

To compute the cost of Sinkhorn iterations, observe that the parameters $\eta$ and $h$ are chosen in such a way that each term on the right-hand side of~\eqref{EQ:final_bound_nash} are smaller than $\eps/2$:
$$
 \left(16+\frac{8}{\gamma^*}\mathfrak{m}\right) \lambda^* h , \left(4\gamma^*+ 2\mathfrak{m}\right) \eta \leq \frac{\eps}{2}
$$
Hence,
$$
\tilde \mu_{\min} \ge \frac{h \lambda^*}{\gamma^*}=\frac{\eps}{48\mathfrak{m}}\,,\quad \text{and} \quad 
\left(\frac{\gamma^*}{h \lambda^*}\right)^3=\left(\frac{48\mathfrak{m}}{\eps}\right)^3\,.
$$
Together with the prescribed value of $\eta$ and since $\gamma^* \leq \mathfrak{m}$ because of Proposition \ref{PR:Gamma}, we get that the total complexity of Sinkhorn iterations is
$$
\cO\left(\left(\frac{\gamma^*}{h \lambda^*}\right)^3\frac{\log\left(1/\tilde \mu_{\min}\right)}{\eta}\right)=\cO\left(\left(\frac{\mathfrak{m}}{\eps}\right)^{4} \log\left(\frac{\mathfrak{m}}{\eps}\right)\right)
$$
Finally, the last step has a total cost of $\cO(n+\eps)$ which is negligible with respect to the combination of previous steps.
\end{proof}

\bigskip 

\noindent{\bf Acknowledgments.} Vianney Perchet acknowledges support from the French National Research Agency (ANR) under grant
number  (ANR-19-CE23-0026 as well as the support grant, as well as from the grant ``Investissements
d’Avenir'' (LabEx Ecodec/ANR-11-LABX-0047). Philippe Rigollet is supported by NSF grants IIS-1838071, DMS-2022448, and CCF-2106377.

\bibliographystyle{plain}
\bibliography{Blotto}

\appendix
\section{The discrete case} \label{SE:Discrete}

During the construction of an approximated solution of the classical Blotto game, we had to resort to some discretization, and we implicitly proved that some discrete random variables were jointly mixable.  Quite unfortunately, this result can not  be directly generalized  to solve the discrete, pure count Blotto game (where Lotto solutions were computed explicitly \cite{hart2008discrete}).

We recall that in the discrete pure count Blotto problem,  the budget of each player $T_A$ and $T_B$ are non-negative integers and that the amount allocated by players to battlefields, $x_{A,i}$ and $x_{B,j}$, are also non-negative integers. Finally, explicit solutions of the Lotto problem are only available in the pure count problem, when $v_{A,i}=v_{B,i}=1$, hence we shall focus on this case. In discrete Lotto/Blotto, the probability of equal forces $x_{A,i}=x_{B,i}$ is positive and if this happens, we will still assume that both players wins the battlefield with probability $1/2$.

 \cite{hart2008discrete} described optimal strategies in the associated Lotto game with the following additional notations. The average budget per battlefield are denoted by $a=T_A/n$ and $b=T_B/n$, the uniform distribution on even integers between $0$ and $2m$ is denoted by  $U_e^m={\sf Unif}\{0,2,...,2m\}$ while $U_o^m={\sf Unif}\{1,3,...,2m-1\}$ is the uniform distribution on even integers between 1 and $2m-1$. Those strategies are described in  \cite[Fig.\ 1]{hart2008discrete}.

He also introduced the term of a ``feasible'' distribution, to indicate that $n$ random variable of that distribution are jointly mixable. He then proved the following
\begin{proposition}
\begin{itemize}
    \item If $T_A=mn$ then $U^m_o$ is "feasible" if and only if $T_A$ and $n$ have the same parity
    \item If $T_A=mn$ then $U^m_e$ is "feasible" if and only if $X_A$ is even
    \item If $T_A=mn+r$ with $1 \leq r \leq n-1$ then $(1-\frac{r}{n})U^m_e+\frac{r}{n}U^{m+1}_o$ is "feasible" 
\end{itemize}
\end{proposition}

As a consequence,  \cite{hart2008discrete} characterized Nash equilibrium of the discrete Blotto problem in the following three cases:
\begin{enumerate} 
\item If $T_A=T_B$ (because of the first and third rows of \cite[Fig.\ 1]{hart2008discrete}) 
\item If $mn < T_B < T_{A} < (m+1)n$ for some $m \in \mathds{N}$ (because of the third row of \cite[Fig.\ 1]{hart2008discrete}) 
\item If $mn=T_B < T_A < (m+1)n$ for some $m \in \mathds{N}$, if $X_B$ is even (because of the fourth row of \cite[Fig.\ 1]{hart2008discrete}) 
\end{enumerate}
Now, let us state the following Proposition~\ref{prop:discretemix} that will  imply the above 3 points. The proof, rather technical (yet algorithmic) is postponed.

\begin{proposition}\label{prop:discretemix}
Discrete random variables $ {\sf Unif}\{0,\ldots,\ell_i\}$ are jointly mixable if and only if their continuous counterparts $ {\sf Unif}[0,\ell_i]$ are jointly mixable and $\sum_i \ell_i$ is even.
\end{proposition}
This proposition allows us to describe solutions of the discrete Blotto game.
\begin{theorem}
Assume $X_B$ is even and $X_A$ has the same parity than $n$, then the optimal strategies of the discrete Blotto game are given by  \cite[Fig.\ 1]{hart2008discrete}.
Moreover those strategies can be computed using the same algorithmic approach as in Proposition~\ref{prop:discretemix}.
\end{theorem}
\begin{proof}
We will prove  the theorem first when $T_A=mn$ for some $m \in \mathds{N}$ and then when $T_B < mn < T_A$, again for some $m\in\mathds{N}$. The other cases are already covered \cite{hart2008discrete}.
\begin{enumerate}
    \item[1.] If $T_B < mn = T_A$, then the optimal strategy of Player $A$ is $U_o^m$, which is feasible if $T_A$ and $n$ have the same parity. Indeed, since  $U_o^m={\sf Unif}\{1,\ldots,2m-1\} = 1 +2{\sf Unif}\{0,\ldots,m-1\}$, the problem reduces to mixing  $n$ random variables of law ${\sf Unif}\{0,\ldots,m-1\}$; this requires that $n(m-1)=T_A-n$ is even.
    
    Player $B$ marginals are $(1-\frac{b}{m})\delta_0+\frac{b}{m}U^m_{o/e}$. So it remains to prove that those random variables are jointly mixable. In particular, this holds if $(1-\frac{b}{m})\delta_0+\frac{b}{m}{\sf Unif}\{0,\ldots,2m\}$ are jointly mixable by choosing  appropriate weights on $U^m_{o}$ and $U^m_{e}$.

    In the continuous case (i.e., if uniform distributions are over $[0,m]$ instead of $\{0,m\}$),  \cite{Rob06} constructed explicit couplings between such random variables by reducing to coupling of uniform continuous random variables. We  apply the exact same techniques, yet we just need to ensure that the intervals of the continuous uniform distributions start and end on integers and that the sums of lengths are always even. But this  immediately happens  as soon as $X_A/n$ is an integer. 
    
    Indeed, all the  couplings introduced in \cite{Rob06} involves uniform random variables over intervals $[\alpha_i,\beta_i]$ where $\alpha_i,\beta_i \in \{0,2m,X_B-2m,bn-4m, ...\}$. As a consequence, all interval lengths are even if $T_B=bn$ is even, which implies that the discrete uniform variables over $\{\alpha_i,\ldots,\beta_i\}$ are jointly mixable.
    
    \item[2.] In the second case, the strategy of player $B$ is to jointly mix $n$ random variables of distribution $(1-\frac{b}{m})\delta_0+\frac{b}{m} {\sf Unif}\{0,2,\ldots,2m\}$ which is equivalent to mixing $n$ distributions $(1-\frac{b}{m})\delta_0+\frac{b}{m} {\sf Unif}\{0,1,\ldots,m\}$. Using again, exactly as  above, the construction based on the continuous couplings of \cite{Rob06}, this is possible as soon as $T_B$ is even.
    
    On the other hand, Player $A$ marginals are $(1-\alpha)U_o^m+\alpha U_o^{m+1}$, thus we need to prove that  $(1-\alpha){\sf Unif}\{0,m-1\}+\alpha {\sf Unif}\{0,m\}$ are jointly mixable. To ensure this, one just need to select $n\alpha$ battlefields at random and to allocate ${\sf Unif}\{0,m\}$ on them (and ${\sf Unif}\{0,m-1\}$ on the  $n-n\alpha$ remaining battlefields). As a consequence, we end up in mixing $n\alpha$ uniforms ${\sf Unif}\{0,m\}$ and $n-n\alpha$  uniforms ${\sf Unif}\{0,m-1\}$, which is possible as soon as  $n\alpha m+(n-n\alpha)(m-1)=T_A-n$ is even, i.e., if $T_A$ and $n$ have the same parity.
\end{enumerate}
These two claims give the result.
\end{proof}

We  finally prove Proposition~\ref{prop:discretemix}

\begin{proof}[Proof of Proposition~\ref{prop:discretemix}]
The proof of the necessary part of the condition is identical to the continuous case. The only difference is the fact that $\sum \ell_i$ must be even. This is naturally  implied by the fact that $\sum X_i$ is always an integer, hence $\sum \mathds{E} X_i = \sum \ell_i/2$ should also be an integer if these random variables are jointly mixable.

It only remains  to prove this statement for $n\geq3$ as the case $n=2$ is trivial.  Moreover, proving the statement for $n>3$ can be be reduced to the case $n=3$  with a simple induction over  $n$, as in the continuous  case \cite{wangJointMixability2016}. Indeed, since ${\sf Unif}\{0,\ldots,\ell_1\}$, ${\sf Unif}\{0,\ldots,\ell_2\}$ and ${\sf Unif}\{0,\ldots,\ell_1+\ell_2\}$ are jointly mixable (because they satisfy the condition of the Proposition for $n=3$), it is possible to reduce the joint mixability of $n$ uniform to only $3$. As a consequence,  we will solely focus on $n=3$ and and we assume wlog that $\ell_1 \leq \ell_2\leq\ell_3$

\medskip

The proof will be based on another induction on the maximal size $\ell_3$. We will distinguish three cases, depending on whether $\ell_1=\ell_2=\ell_3$, or $\ell_1<\ell_2 = \ell_3$, or $\ell_2<\ell_3$.

\medskip

\textbf{First case:} $\ell_1=\ell_2=\ell_3$ (in particular, this implies that $\ell_1,\ell_2,\ell_3$  are  even since $\ell_1+\ell_2+\ell_3=3\ell_i$ is even by assumption).

If $\ell_1=2$, then the uniform coupling on the triplets $(0,1,2)$, $(1,2,0)$ and $(2,0,1)$ proves joint mixability. If $\ell_1=4$, then considering the uniform  distribution on the following set
\begin{align*}\Big\{(3,3,0);  (3,0,3); (0,3,3); (4,1,1); (1,4,1); (1,1,4); \\ (0,2,4); (2,4,0); (4,0,2); (2,2,2)\Big\}\end{align*} is sufficient to prove joint mixability. From now on, we shall assume that $\ell_1=\ell_2=\ell_3 \geq 6$.

Consider the coupling defined by \begin{align*}
    X^{1,-}_1 \sim \delta_0, \quad X^{1,-}_2\sim {\sf Unif}\{\frac{\ell_1}{2}+1,\ell_1-1\} \ \ \text{ and } \\  X^{1,-}_3=\frac{3}{2}\ell_1-X_2^{1,-} \sim {\sf Unif}\{\frac{\ell_1}{2}+1,\ell_1-1\}.\end{align*} Then it immediately follows that $X^{1,-}_1+X^{1,-}_2+X^{1,-}_3 = \frac{3}{2}\ell_1$. Similarly, this property holds with the following alternative coupling   \begin{align*} X^{1,+}_1 \sim \delta_{\ell_1}, \quad X^{1,+}_2\sim {\sf Unif}\{1,\frac{\ell_1}{2}-1\}\ \ \text{ and } \\  X^{1,+}_3=\frac{1}{2}\ell_1-X_2^{1,+} \sim {\sf Unif}\{1,\frac{\ell_1}{2}-1\}.\end{align*} We define similarly the coupling $(X_i^{2,\pm})_{i \in [3]}$ and $(X_i^{3,\pm})_{i \in [3]}$ where the role of $X_1$ is exchanged with $X_2$ and, respectively, $X_3$. We also define the last coupling $X_i^0=\delta_{\frac{\ell_i}{2}}$.

Then $X_i$ can be  decomposed as follows: $$X_i = \frac{1}{\ell_1+1} \Big\{\sum_{k\in [3]} X_i^{k,-}+\sum_{k\in [3]}X_i^{k,+}+\frac{2}{\frac{\ell_i}{2}-1}X_i^0\Big\} + \frac{\ell_1-5}{\ell_1+1}Y_i,$$
where $Y_i \sim {\sf Unif}\{1,\ldots, \ell_1\}$. A simple induction gives the joint mixability of $\{X_1,X_2,X_3\}$

\bigskip

\textbf{Second case:} $\ell_1<\ell_2=\ell_3$ (in particular, this implies that $\ell_1$ is even). There are two specifics cases $\ell_1=2$ and $\ell_2\in\{3,4\}$ that are constructed explicitly as follows.

If $\ell_1=2$ and $\ell_2=\ell_3=3$, a joint mixability coupling is 
$$
\frac{1}{6}\Big(\delta_{(0,3,1)}+\delta_{(0,1,3)}+\delta_{(2,2,0)}+\delta_{(2,0,2)}\Big) +\frac{1}{12}\Big(\delta_{(1,3,0)}+\delta_{(1,0,3)}+\delta_{(1,1,2)}+\delta_{(1,2,1)}\Big)
$$
while if $\ell_1=2$ and $\ell_2=\ell_3=4$, a valid coupling is
\begin{align*}
    \frac{2}{15}\Big(\delta_{(2,3,0)}+\delta_{(2,0,3)}+\delta_{(0,4,1)}+\delta_{(0,1,4)}+\delta_{(1,2,2)}\Big)\\
    +\frac{1}{15}\Big(\delta_{(1,0,4)}+\delta_{(1,4,0)}+\delta_{(0,3,2)}+\delta_{(2,2,1)}+\delta_{(1,1,3)}\Big)
\end{align*}

For the other cases, we  consider similar couplings as above, i.e.,
\begin{align*}
    X^{3,-}_3\sim\delta_0,\quad X^{3,-}_2\sim 
{\sf Unif}\{\ell_2-\frac{\ell_1}{2},\ell_2-1\}
 \ \ \text{ and } \\   X_1^{3,-}= \ell_2+\frac{\ell_1}{2}-X_2^{3,-} \sim {\sf Unif}\{\frac{\ell_1}{2}+1,\ell_1\}\end{align*}
and also
\begin{align*}X^{3,+}_3\sim\delta_{\ell_3}, \quad X^{3,-}_2\sim 
{\sf Unif}\{1,\frac{\ell_1}{2}\}
 \ \ \text{ and } \\   X_1^{3,+}= \frac{\ell_1}{2}-X_2^{3,-} \sim {\sf Unif}\{0,\frac{\ell_1}{2}-1\}.\end{align*}
We define $X_i^{2,\pm}$ similarly.

If $\ell_1 \geq 4$, we introduce the following random variables (the case $\ell_1=2$ is detailed just below)
\begin{align*}
Y_1^{1}\sim\delta_{\frac{\ell_1}{2}}, Y_1^{2}\sim{\sf Unif}\{0,\ell_1\}, Y_1^{3}\sim{\sf Unif}\{0,\ell_1\} \\  \text{ and } \ Y_1^{4}\sim{\sf Unif}\{0,\ell_1\}
\end{align*}
and similarly 
\begin{align*}
Y_2^{1}\sim{\sf Unif}\{1,\ell_2-1\}, Y_2^{2}\sim{\sf Unif}\{\frac{\ell_1}{2}+1,\ell_2-\frac{\ell_1}{2}-1\}, Y_2^{3}\sim{\sf Unif}\{1,\ell_2-1\} \\ \text{ and } \  Y_2^{4}\sim{\sf Unif}\{1,\ell_2-1\}
\end{align*}
and 
\begin{align*}
Y_3^{1}\sim{\sf Unif}\{1,\ell_3-1\}, Y_3^{2}\sim{\sf Unif}\{1,\ell_3-1\}, Y_3^{3}\sim{\sf Unif}\{\frac{\ell_1}{2}+1,\ell_2-\frac{\ell_1}{2}-1\} \\ \text{ and } \  Y_3^{4}\sim{\sf Unif}\{1,\ell_3-1\}
\end{align*}
So that we can decompose 
$$
X_i = \frac{1}{\ell_3+1}(X_i^{3,-}+X_i^{3,+}+X_i^{2,-}+X_i^{2,+})+p_1Y_i^1+p_2Y_i^2+p_3Y_i^3+p_4^i
$$
where the probability are defined by
$$
p_1=\frac{4}{\ell_3+1}\frac{1}{\ell_1}, p_3=p_2=\frac{2}{\ell_3+1}\frac{\ell_2-\ell_1-1}{\ell_1} \ \text{ and } p_4= 1-\frac{4}{\ell_3+1}-p_1-p_2-p_3
$$
Those couplings are well defined and satisfy the theorem length condition, hence we will be able to proceed by induction.

It remains to consider the case where $\ell_1=2$ (plugging $\ell_1=2$ in the above construction would give $p_4<0$ which is obviously impossible).

We define the variables 
$$
Z_1^{1}\sim\delta_{1}\  \text{ and } \ Z_1^{2}\sim{\sf Unif}\{0,2\} 
$$
and similarly 
$$
 Z_2^{1}\sim{\sf Unif}\{2,\ell_2-2\}, \ \text{ and } \  Z_2^{2}\sim{\sf Unif}\{2,\ell_2-2\}
$$
and 
$$
  Z_3^{1}\sim{\sf Unif}\{2,\ell_2-2\} \ \text{ and } \  Z_3^{2}\sim{\sf Unif}\{2,\ell_3-2\}
$$
So that we can decompose 
$$
X_i = \frac{1}{\ell_3+1}(X_i^{3,-}+X_i^{3,+}+X_i^{2,-}+X_i^{2,+})+q_1Z_i^1+q_2Z_i^2
$$
where the probability are defined by
$$
q_1=\frac{2}{\ell_3+1} \ \text{ and } q_2= 1-\frac{4}{\ell_3+1}-q_1=1-\frac{6}{\ell_3+1},
$$
and the result also follows by induction.
\bigskip

\textbf{Third case:} $\ell_1<\ell_2<\ell_3$. We are going to proceed by induction (on the maximal length) as before, and we consider the following couplings
\begin{align*}
X_3^{-} = \delta_0, \quad X_2^{-}\sim {\sf Unif}\{\frac{\ell_3+\ell_2-\ell_1}{2},\ell_2\} \\ \text{and} \ X_1^{-} = \frac{\ell_3+\ell_2+\ell_1}{2} - X_2^{-} = {\sf Unif}\{0,\frac{\ell_1+\ell_2-\ell_3}{2}\}
\end{align*}
and the similar one
\begin{align*}
X_3^{-} = \delta_{\ell_3},\quad X_2^{-}\sim {\sf Unif}\{0,\frac{\ell_1+\ell_2-\ell_3}{2})\} \\ \text{and} \ X_1^{-} = \frac{\ell_3+\ell_2+\ell_1}{2} - X_2^{-} = {\sf Unif}\{\frac{\ell_1+\ell_3-\ell_2}{2},\ell_1\}
\end{align*}
It remains to introduce the following random variables
\begin{align*}
Y_1^{1}\sim{\sf Unif}\{\frac{\ell_1+\ell_2-\ell_3}{2}+1,\frac{\ell_1+\ell_3-\ell_2}{2}-1\},\quad Y_1^{2}\sim{\sf Unif}\{0,\ell_1\} \\ \text{ and } \ Y_1^{3}\sim{\sf Unif}\{0,\ell_1\}
\end{align*}
and similarly 
\begin{align*}
Y_2^{1}\sim{\sf Unif}\{0,\ell_2\},\quad Y_2^{2}\sim{\sf Unif}\{\frac{\ell_1+\ell_2-\ell_3}{2}+1,\frac{\ell_2+\ell_3-\ell_1}{2}-1\} \\ \text{ and } \ Y_2^{3}\sim{\sf Unif}\{0,\ell_2\} 
\end{align*}
and 
\begin{align*}
Y_3^{1}\sim{\sf Unif}\{1,\ell_3-1\},\quad Y_3^{2}\sim{\sf Unif}\{1,\ell_3-1\} \\ \text{ and } \ Y_3^{3}\sim{\sf Unif}\{1,\ell_3-1\}.
\end{align*}
So that we can decompose 
$$
X_i = \frac{1}{\ell_3+1}(X_i^-+X_i^+)+p_1Y_i^1+p_2Y_i^2+p_3Y_i^3
$$
where the probability are defined by
\begin{align*}
p_1=\frac{1}{\ell_3+1}\frac{\ell_3-\ell_2-1}{\frac{\ell_1+\ell_2+\ell_3}{2}+1}, \quad p_2=\frac{1}{\ell_3+1}\frac{\ell_3-\ell_1-1}{\frac{\ell_1+\ell_2+\ell_3}{2}+1} \\ \text{ and } \ p_3=1-\frac{2}{\ell_3+1}-p_1-p_2
\end{align*}
The proof relies on a simple induction by noticing that $\{Y_1^1,Y_2^1,Y_3^1\}$ satisfies the theorem condition (as soon as $p_1 >0$) since these three random variables are uniform over intervals of respective lengths $\ell_3-\ell_2-2$, $\ell_2$ and $\ell_3-2$ (in particular, $\ell_2$ is the maximum of these 3 quantities, necessarily $\ell_2=\ell_3-1$ and $p_1=0$). Similarly $\{Y_1^2,Y_2^2,Y_3^2\}$ and $\{Y_1^3,Y_2^3,Y_3^3\}$ satisfy the theorem condition as well.
 \end{proof}
 \section{Omitted Proofs and algorithms}

 \subsection{Algorithm \LOTTOTOBLOTTO}\label{SE:LOTTO2BLOTTO}
 We provide the simple pseudo-code of the main algorithm; it is decomposed into several procedures described in subsequent sections.
 \begin{center}
    \begin{algorithm}
     \caption{\LOTTOTOBLOTTO \label{Algo:Lotto2Blotto}}
\SetAlgoLined
\KwData{Length \& weight vectors $\mathbf{b} \in \R_+^n$ and  $\bp \in [0,1]^n$, budget $T$, approximation levels $\eta, h$\;}
\KwResult{Allocation vector $X \in \R_+^n$\;}

$\{(\bp^{(j)},q_j)\}_j \leftarrow \DECOMPOSE(\mathbf{b},T,\mathbf{p})$\Comment{Couplings with one strict mixture}

Sample $j^*=j$ with probability $q_j$\;
$(\cI_0,\cI_1,\cI_2,\cI_3,\cI_4)\leftarrow \REDUCTION(\mathbf{b},\bp^{(j^*)})$\Comment{Reduction to 4 random variables}

\For{$i \in [4]$}{$b^*_i\leftarrow\sum_{j\in\cI_j}b_j$; $p^*_j\leftarrow\sum_{j\in\cI_j}p_j$\;}
$(\mu_1,\mu_2,\mu_3,\mu_4) \leftarrow \DISCRETIZE((b^*_1,b^*_2,b^*_3,b^*_4),p^*_4,h)$\Comment{Appropriate Discretization}

$(\xi_1,\xi_2,\xi_3,\xi_4)\leftarrow \SINKHORN((\mu_1,\mu_2,\mu_3,\mu_4),\eta)$\Comment{Numerical computations}

$X \leftarrow \SAMPLE((\xi_1,\xi_2,\xi_3,\xi_4),T,\mathbf{b},(\cI_0,\cI_1,\cI_2,\cI_3,\cI_4),h)$\Comment{Reconstruction and sampling}

 \textbf{Output:} $X$ \Comment{A valid allocation}
\end{algorithm}
\end{center}
 \newpage
 \subsection{Proof of  Proposition \ref{PR:Gamma} and associated Algorithm~\LOTTO}\label{SE:ProofOfPR:Gamma}
  \begin{center}
    \begin{algorithm}
\SetAlgoLined
\KwData{Battlefield values $v_{A,i}$ and $v_{B,i}$, budgets $T_A$ and $T_B$\;}
\KwResult{Vector $\mathbf{b}=(b_i)_{i \in [n]}$ and probability vector $\bp \in [0,1]^n$\;}
$\gamma\leftarrow 0$, $\lambda \leftarrow 0$, $k\leftarrow 0$\Comment{Initialization}

Sort $\big\{\frac{v_{A,i}}{v_{B,i}}; i \in [n]\big\}$\Comment{Sort to get degree 3 polynomial equation to solve}

\While{$\gamma=0$}{
$\mathcal{S}_k \leftarrow \text{roots}\Big( \gamma T_A\sum_{i=1}^kv_{a,i} + \gamma^3 T_A\sum_{i=k+1}^n \frac{v_{B,i}^2}{v_{A,i}}  -T_B\sum_{i=1}^k \frac{v^2_{A,i}}{v_{B,i}}-T_B\gamma^2\sum_{i=k+1}^n v_{B,i}\Big)$\;
\For{$g \in \mathcal{S}_k$}{$\gamma\leftarrow g\mathds{1}\big\{\gamma\in [\frac{v_{A,k}}{v_{B,k}},\frac{v_{A,k+1}}{v_{B,k+1}}]\big\}$\Comment{Keep valid root of degree 3 polynomial}}
$k\leftarrow k+1$\Comment{Move to next interval between two ratios $\frac{v{A,k}}{v_{B,k}}$}
}
$\lambda \leftarrow \frac{1}{T_A}\frac12\sum_{i=1}^n(\gamma v_{B,i})\wedge \frac{(v_{A,i})^2}{\gamma v_{B,i}}$\;
\For{$i\in[n]$}{$b_i \leftarrow \frac{\gamma v_{B,i}}{\lambda}\wedge\frac{v_{A,i}}{\lambda}$,\ 
$p_i \leftarrow \frac{v_{A,i}}{\gamma v_{B,i}}\wedge 1$\Comment{Description of Lotto strategies}}
 \textbf{Output:} $(\mathbf{b},\bp)$\Comment{lenghts of uniforms and associated weights}
 \caption{\LOTTO \label{Alg:Lotto}}
\end{algorithm}
\end{center}

 \begin{proof} 
 First, reorder the battlefields by increasing reward ratios $v_{A,i}/v_{B,i}$ so  we can assume that 
    $$\frac{v_{A,1}}{v_{B,1}} \leq \frac{v_{A,2}}{v_{B,2}}\leq \cdots \leq \frac{v_{A,n}}{v_{B,n}}
    $$

Recall that $f(\gamma)$ is the left-hand side of~\eqref{EQ:defGamma}. We first observe that $f$ is continuous on $\R$. Indeed, $f$ is obviously continuous on each open interval 
$$
\left(\frac{v_{A,i}}{v_{B,i}}, \frac{v_{A,i+1}}{v_{B,i+1}}\right)
$$
To check that it is continuous at $\gamma_{i_0}:=v_{A,i_0}/v_{B,i_0}$, note that for $\gamma$ in a small enough neighborhood of $\gamma_{i_0}$, we have
$$
\cN(\gamma)=\left\{
\begin{array}{ll}
 \{i_0 + 1,\ldots,  n\}    & \text{ if } \gamma>\gamma_{i_0} \\
 \{i_0, \ldots, n\}    & \text{ if } \gamma\le \gamma_{i_0} \,.
\end{array}
\right.
$$
Hence $f$ is left-continuous at $\gamma_{i_0}$. Moreover, we can check right-continuity by observing that
\begin{align*}
     \lim_{\substack{\gamma  \to \gamma_{i_0}\\ \gamma > \gamma_{i_0} }} f(\gamma)=
\gamma_{i_0}^3\left(T_A \sum_{i=i_0+1}^n\frac{v_{B,i}^2}{v_{A,i}}\right)-\gamma_{i_0}^2T_B\sum_{i=i_0+1}^n v_{B,i}+ \gamma_{i_0} T_A\sum_{i=1}^{i_0}v_{A,i}-T_B\sum_{i=1}^{i_0} \frac{v_{A,i}^2}{v_{B,i}}=f(\gamma_{i_0})
\end{align*}
where the last identity can be readily checked by substitution.

Having proved that $f$ is continuous, note that for $\gamma$ large enough we have that $\cN(\gamma)=\emptyset$ so that 
\begin{equation*}
    \label{EQ:defg}
    f(\gamma)=\gamma T_A\sum_{i=1}^nv_{A,i}-T_B\sum_{i=1}^n \frac{v_{A,i}^2}{v_{B,i}} \xrightarrow[\gamma\to \infty]{}+\infty\,.
\end{equation*}
Moreover, for $\gamma>0$ small enough, we have $\cN(\gamma)=[n]$ so that 
\begin{equation*}
    \label{EQ:defh}
    \frac{f(\gamma)}{\gamma^2}=\gamma\left(T_A \sum_{i=1}^n\frac{v_{B,i}^2}{v_{A,i}}\right)-T_B\sum_{i =1}^n v_{B,i}\xrightarrow[\gamma\to 0]{}-T_B\sum_{i =1}^n v_{B,i}<0\,.
\end{equation*}
Hence, by the intermediate value theorem, there exists $\gamma^*>0$ such that $f(\gamma^*)=0$. Moreover, observe that if one sets $\cN(\gamma)=[N]$ for $N\in\{0,\dots,n\}$, with the convention that $[0]=\emptyset$, Equation~\eqref{EQ:defGamma} becomes a polynomial equation of degree three with at most three solutions denoted $\gamma_{N,i}, i=1,2,3$. Hence, $\gamma^* \in \bigcup_{N=0}^n \{\gamma_{N,1}, \gamma_{N,2}, \gamma_{N,3}\}$ can take at most $3n+3$ values.

\medskip

We now check the bounds on the possible values of $\gamma^*$. To that end, recall from~\eqref{EQ:defGamma} that
$$
f(\gamma)=\sum_{i=1}^n\left(\gamma T_A-T_B\frac{v_{A,i}}{v_{B,i}}\right)\cdot \left(\gamma^2\frac{v_{B,i}^2}{v_{A,i}}\wedge v_{A,i}\right) \,. 
$$

To prove the an upper bound on $\gamma^*$, observe that
$$
\gamma T_A-T_B\frac{v_{A,i}}{v_{B,i}} \ge 0
 \iff \gamma \ge \frac{T_B}{T_A} \frac{v_{A,i}}{v_{B,i}}
 \iff 
 \gamma^2\frac{v_{B,i}^2}{v_{A,i}}\wedge v_{A,i} \ge \left(\frac{T_B}{T_A}\right)^2v_{A,i} \wedge v_{A,i}= \left(\frac{T_B}{T_A}\right)^2v_{A,i}\,,
$$
where the last identity follows from the assumption that $T_B \le T_A$.

It yields
$$
f(\gamma)\ge \left(\frac{T_B}{T_A}\right)^2\sum_{i=1}^n\left(\gamma T_A-T_B\frac{v_{A,i}}{v_{B,i}}\right)v_{A,i}=\left(\frac{T_B}{T_A}\right)^2\cdot \left(\gamma T_A  -T_B(1+ \chi^2(v_A \|v_B)\right)
$$
Therefore, if 
$$
\gamma>\frac{T_B}{T_A}(1+\chi^2(v_A \|v_B))\,,
$$
then $f(\gamma)>0$, which yields the desired upper bound on $\gamma^*$.

To prove the lower bound on $\gamma^*$, we proceed essentially in the same fashion:
$$
\gamma T_A-T_B\frac{v_{A,i}}{v_{B,i}} \le 0
\implies \gamma \le \frac{T_B}{T_A} \frac{v_{A,i}}{v_{B,i}}
\implies \gamma^2\frac{v_{B,i}^2}{v_{A,i}}\wedge v_{A,i} =\gamma^2\frac{v_{B,i}^2}{v_{A,i}}\,.
$$

It yields
$$
f(\gamma)\le \gamma^2\sum_{i=1}^n\left(\gamma T_A-T_B\frac{v_{A,i}}{v_{B,i}}\right)\frac{v_{B,i}^2}{v_{A,i}}=\gamma^2 \left((1+\chi^2(v_B\|v_A))\cdot \gamma T_A -T_B\right)
$$
Therefore, if 
$$
\gamma<\frac{T_B}{T_A}\frac{1}{1+\chi^2(v_B \|v_A)}\,,
$$
then $f(\gamma)<0$, which yields the desired lower bound on $\gamma^*$.

To complete the proof of the proposition, it remains to observe that the computational complexity is dominated by sorting reward ratios which costs $\cO(n \log n)$ operations since finding roots of degree three polynomials for each of $\cO(n)$ polynomials costs $\cO(1)$ time.

e if $X_B = X_A$ and $\frac{v_{A,1}}{v_{B,1}}=\frac{v_{A,n}}{v_{B,n}}$, which implies that $v_{A,i}=v_{B,i}$ (as they both sum to 1).

 \end{proof}
 
 \subsection{Proofs of Proposition \ref{PR:JointMixability} and Theorem \ref{TH:LottotoBlotto}}\label{SE:ProofJointMixability}
 
 \subsubsection{Proof of Proposition \ref{PR:JointMixability}}
\begin{proof} 
Note first that~\eqref{EQ:JMcond1} is necessary. Indeed,  denote by $\max \in [k]$ any index such that $b_{\max}=\max_ib_i$ and assume $Z_1, \ldots, Z_k$ are coupled so that
$$
\sum_{i\in [k]}Z_i=\sum_{i\in [k]}\E[Z_i]=\frac12\sum_{i\in [k]}p_ib_i \quad \text{a.s.}
$$
We have for $\delta>0$ small enough, 
$$
p_{\max}\delta=\p(Z_{\max}>b_{\max}(1-\delta))\le \p\big(\sum_{i\in [k]}Z_i >b_{\max}(1-\delta)\big)=\p\big(\frac{1}{2}\sum_{i\in [k]}p_ib_i >b_{\max} (1-\delta)\big)\,.
$$
Hence
$$
\frac{1}{2}\sum_{i\in [k]}p_ib_i >b_{\max} (1-\delta)
$$
since the above deterministic inequality holds with positive probability, it holds with probability one for all $\delta>0$. Letting $\delta \to 0$ yields~\eqref{EQ:JMcond1}.

To show that~\eqref{EQ:JMcond1} is sufficient, recall from~\cite[Theorem~3.2]{wangJointMixability2016} that a collection of random variables $Z_1,\ldots,Z_k$ where  $Z_i$ is a continuous random variable, with  non-increasing density function and supported on the interval $[0,b_i]$ are jointly mixable  if and only if
\begin{equation*}
\max_{i \in [k]} b_i \leq  \sum_{i \in [k]} \mathds{E}[Z_i] \leq \sum_{i \in [k]} b_i - \max_{i \in [k]} b_i\,.
\end{equation*}
In particular, if $\E[Z_i]=p_ib_i/2$, these two inequalities reduce to~\eqref{EQ:JMcond1}.

Unfortunately the $Z_i$s do not have a  density so we use the following approximation. For $i \in [k]$, let $\eps<b_i$ and 
$$
Z_i^\eps \sim (1-q_i^\eps){\sf Unif}[0,\eps] + q^\eps_i {\sf Unif}[0,b_i]\,,
$$
where $q_i^\eps \in (0,1)$ is chosen precisely so that $Z_i^\eps$ has the same expectation and the same support as $Z_i$. Moreover, $Z_i^\eps$ has a monotone decreasing density and hence the $Z_i^\eps$ are jointly mixable under condition~\eqref{EQ:JMcond1} resulting in a coupling $\pi^\eps$ over the product space $\prod_{i \in [k]} [0,b_i]$. By Prokhorov's theorem, letting $\eps\to 0$ implies that
 $Z_1, \ldots, Z_k$ are jointly mixable.

\end{proof}
 
 \subsubsection{Proof of Theorem \ref{TH:LottotoBlotto}}
 \begin{proof}
 Recall that the marginal strategy  $Z_i$  of player $A$ on battlefield $i$ has a distribution of the form~\eqref{EQ:Zimixt} with
\begin{align*}
    p_i&=1, & b_i&=\frac{\gamma^* v_{B,i}}{\lambda^*}, & \text{if } i &\in \cN(\gamma^*)\\
    p_i&=\frac{v_{A,i}}{\gamma^*v_{B,i}}, & b_i&=\frac{ v_{A,i}}{\lambda^*}, & \text{if } i &\notin \cN(\gamma^*)
\end{align*}
Plugging this values into~\eqref{EQ:JMcond1} yields
$$
\max_{i \in \cN(\gamma^*)} \frac{\gamma^* v_{B,i}}{\lambda^*}\vee \max_{i \notin \cN(\gamma^*)} \frac{ v_{A,i}}{\lambda^*} \le \frac12\left(\sum_{i \in \cN(\gamma^*)}\frac{\gamma^* v_{B,i}}{\lambda^*} + \sum_{i \notin \cN(\gamma^*)} \frac{v_{A,i}^2}{\lambda^*\gamma^*v_{B,i}}\right)=T_A\,,
$$
where the last equality follows from the saturation of the budget constraint in~\eqref{EQ:Lambda1}. It is easy to check that
$$
\max_{i \in \cN(\gamma^*)} \gamma^* v_{B,i}\vee \max_{i \notin \cN(\gamma^*)}  v_{A,i}=\max_{i \in [n]} (\gamma^* v_{B,i} \wedge v_{A,i})
$$

Similarly, we get using~\eqref{EQ:Lambda2} that
$$
\max_{i \in [n]} (\gamma^* v_{B,i} \wedge v_{A,i}) \le \lambda^*T_B\,,
$$
is a necessary and sufficient condition for mixability the Lotto strategy of player $B$ into a Blotto strategy.
The proof can then be concluded by recalling that $T_A \ge T_B$.

\end{proof}

 \subsection{Proof of Corollary \ref{cor:stability}}
 \label{SE:Proofcor:stability}
\begin{proof}
Using~\eqref{EQ:Lambda2}, we get
$$
\frac{\lambda^*T_B}{\gamma^*}=\frac1{2\gamma^*}\sum_{i=1}^n\frac{(\gamma^* v_{B,i})^2}{v_{A,i}}\wedge v_{A,i} \ge  \frac{1}{2}\left(\gamma^*\wedge \frac{1}{\gamma^*}\right)\sum_{i=1}^n\frac{ v_{B,i}^2}{v_{A,i}}\wedge v_{A,i}\,.
$$
Next, observe that
\begin{align*}
   \sum_{i=1}^n\frac{ v_{B,i}^2}{v_{A,i}}\wedge v_{A,i}&=\sum_{i=1}^n\frac{ v_{B,i}^2}{v_{A,i}}\1(v_{A,i}>v_{B,i})+\sum_{i=1}^nv_{A,i}\1(v_{A,i}\le v_{B,i})\\
   &=\sum_{i=1}^n\frac{ v_{B,i}^2}{v_{A,i}}\1(v_{A,i}>v_{B,i})+1-\sum_{i=1}^nv_{A,i}\1(v_{A,i}> v_{B,i}) \\
   &=1-\sum_{i=1}^n v_{A,i}\left(1-\left(\frac{v_{B,i}}{v_{A,i}}\right)^2\right)\1(v_{A,i}> v_{B,i})\,.
\end{align*}
Noting now that $1-x^2\le 2-2x$ for $x\in [0,1]$, we get for $x=v_{B,i}/v_{A,i}$ that
\begin{align*}
    \sum_{i=1}^n\frac{ v_{B,i}^2}{v_{A,i}}\wedge v_{A,i}&\ge 1-\sum_{i=1}^n v_{A,i}\left(1-\frac{v_{B,i}}{v_{A,i}}\right)\1(v_{A,i}> v_{B,i})=1-{\sf TV}(v_{A}, v_{B})\,,
\end{align*}
where ${\sf TV}$ denotes the total variation distance and is defined as
$$
{\sf TV}(v_A, v_B)=\frac12\sum_{i=1}^n |v_{A,i}-v_{B,i}|\,.
$$
Using Pinsker's and Jensen's inequalities, see, e.g.,\cite[Chapter~2]{Tsy09}, we get the classical result
$$
{\sf TV}(v_{A}, v_{B}) \le \frac{1}{2}\sqrt{\chi^2(v_B\|v_A)}\le \frac{r}2
$$

Next, note that Proposition~\ref{PR:Gamma}  yields
$$
\gamma^* \wedge \frac1{\gamma^*} \ge  \frac{T_B}{T_A}\frac{1}{1+\chi^2(v_B\|v_A) }\wedge \frac{T_A}{T_B}\frac{1}{1+\chi^2(v_A\|v_B) }\ge \frac{T_B}{T_A}\frac{1}{1+r^2}\,.
$$
Hence we have established that
\begin{align*}
    \frac{\lambda^*T_B}{\gamma^*} &\ge \frac{T_B}{2T_A}\frac{1-r/2}{1+r^2}\ge\frac{T_B}{2T_A}(1-r)\,.
   \end{align*}
Together with~\eqref{EQ:blotto_cond_mix}, this completes our proof.

\end{proof}

\subsection{Proof of Lemma \ref{LEM:reduc_single_mixture} and Algorithm \DECOMPOSE}
\label{SE:ProofLEM:reduc_single_mixture}
\begin{center}
    \begin{algorithm}
\SetAlgoLined
\KwData{length vector $\mathbf{b} \in \R_+^n$, budget $T$, weights vector $\bp=(p_i)_{i \in [n]} \in [0,1]^n$\;}
\KwResult{(at most) $n$ vectors $\bp^{(k)} \in [0,1]^n$ and weights $q_k \in [0,1]$\;}
$k\leftarrow 0$\Comment{Initialization}

\While{$\supp(\bp)>2$}{
$k\leftarrow k+1$\;
$(\bp^{(k)},\bar{\bp}^{(k)},\theta_k) \leftarrow \Cube(\mathbf{b},T,\bp)$\Comment{Find an extreme point $\bp^{(k)}$}

$\bp\leftarrow \bar{\bp}^{(k)}$\;
$q_k\leftarrow q_{k-1}\frac{\theta_k}{\theta_{k-1}}(1-\theta_{k-1})$\Comment{Compute convex weight of $\bp^{(k)}$}
}
 \textbf{Output:} $\{(\bp^{(j)},q_j)\,; j \in [k]\}$\Comment{(a coupling with  its convex weight)}
 \caption{\DECOMPOSE \label{Algo:Decomp}}
\end{algorithm}
\end{center}
\begin{proof}
Throughout this proof we write 
$$
p_i=\left\{
\begin{array}{cl}
1, & \text{if } i \in \cN(\gamma^*)\\
\frac{v_{A,i}}{\gamma^* v_{B,i}},& \text{if } i \notin \cN(\gamma^*)
\end{array}\right.
$$
so that the $i$th marginal of $\pi$ is given by 
$$
(1-p_i)\delta_0+p_i {\sf Unif}\left[0, \frac{\gamma^*v_{B,i}\wedge v_{A,i}}{\lambda^*}\right]\,.
$$

Note that~\eqref{EQ:mixcouplings} consists in representing $\pi$ as a convex combination of couplings. To obtain this representation we are going to appeal to Carath\'eodory's theorem. However the latter requires finite dimension so we first make the following observation: the map $\bp^{(k)}=(p_1^{(k)}, \ldots, p_n^{(k)}) \mapsto \pi_k$, (resp. $\bp=(p_1, \ldots, p_n) \mapsto \pi$) is linear and injective. Therefore,~\eqref{EQ:mixcouplings} is sufficient to produce the decomposition
\begin{equation}
    \label{EQ:carat}
    \bp = \sum_k q_k \bp^{(k)}\,.
\end{equation}
Furthermore, we need $\pi_k$ to be satisfy the budget constraint~\eqref{EQ:const_lotto} of the Lotto game. The saturation of this constraint translates into the constraint $\bp^{k} \in \cH, k=1, \ldots, n$, where $\cH$ is the affine hyperplane in $\R^n$ defined by
\begin{equation}
    \label{EQ:hyperplane}
    \cH=\left\{x \in \R^n\,:\,\sum_{i=1}^n x_i \frac{\gamma^*v_{B,i}\wedge v_{A,i}}{2\lambda^*}=T_A\right\}
\end{equation}

As a result, we must ensure that $\bp^{(k)} \in \cC_n=\cH \cap [0,1]^n$; see Figure~\ref{FIG:carat} for a representation of this constraint set. Since this set has dimension $n-1$, Carath\'edory's theorem ensures the existence of a decomposition~\eqref{EQ:carat} where $\bp^{(k)}$ are extreme points of $\cC_n$. In particular, each such extreme point has at most one coordinate in $(0,1)$; this completes the proof of point {\it 1}.

Note that~\eqref{EQ:carat} readily ensures that $\pi$ is a solution for the Lotto game. Moreover, since $\bp^{(k)} \in \cH$ by construction, we have that marginals of $\pi_k$ automatically satisfy the mixability condition. Hence, we can choose $\pi_k$ to be a joint mix; this completes the proof of point {\it 2.} As stated above, this readily implies that $\pi$ defined in~\eqref{EQ:mixcouplings} is a joint mix and hence a solution for the Blotto game; this completes the proof of point {\it 3.}

It remains to find an efficient algorithm 
that outputs decomposition~\eqref{EQ:mixcouplings}. The algorithm is initialized at $\bp\in \cC_n=\cH \cap [0,1]^n$. From there, we use Lemma~\ref{Lemma:hypercube} to construct two points $\bp^{(1)}, \bar \bp^{(1)} \in \cH \cap \partial [0,1]^n$ such that $\bp=\theta_1 \bp^{(1)}+ (1-\theta_1) \bar \bp^{(1)}$, for some $\theta_1 \in [0,1]$, where $\bp^{(1)}$ is an extreme point of $\cC_n$ and $\bar \bp^{(1)}$ belongs to a face $\cF_{n_1}$ of dimension $n_1<n$. We repeat this procedure as follows.
Define  by $\cC_{n_1}=\cH\cap \cF_{n_1}$ so that  $\bar \bp^{(1)} \in \cC_{n_1}$. Hence, using again Lemma~\ref{Lemma:hypercube}, it may be decomposed as $\bar \bp^{(1)} =\theta_2 \bp^{(2)} + (1-\theta_2) \bar \bp^{(2)}$ for some $\theta_2 \in [0,1]$ and  $\bp^{(2)}$ is an extreme point of $\cC_n$ while $\bar \bp^{(2)} \in \cC_{n_2}$, where $\cC_{n_2}$ is the intersection of $\cH$ with a face $\cF_{n_2}$ of $[0,1]^n$ of dimension $n_2 < n_1$. Moreover, 
$$
\bp=\theta_1 \bp^{(1)} + (1-\theta_1)\theta_2 \bp^{(2)} + (1-\theta_1)(1-\theta_2)\bar \bp^{(2)}
$$
Iterating this procedure yields the decomposition
$$
\bp= \sum_{i=1}^{N+1} \theta_i \prod_{j \le i-1}(1-\theta_j) \bp^{(i)}\,,
$$
with the convention that $\theta_{N+1}=1$ and 
where $N\le n$ and the $\bp^{(j)}$s are extreme points of $\cC_n$. Moreover, one can readily check that for any sequence $\theta_1, \ldots \theta_N \in [0,1]$, $\theta_{N+1}=1$, it holds 
$$
\sum_{i=1}^{N+1}  \theta_i \prod_{j \le i-1}(1-\theta_j)=1\, ,
$$
which gives the result by defining $q_k = \theta_k \prod_{j \le k-1}(1-\theta_j)$.

As we appealed at most $n$ times to Lemma~\ref{Lemma:hypercube}, the overall complexity of this algorithm is of order $\mathcal{O}(n^2\log n)$.
\end{proof}

\begin{figure}
\begin{center}
\begin{tikzpicture}[scale=3]
\draw[black, thick] (0,0,1) -- (1,0,1) -- (1,1,1) -- (0,1,1) -- node[anchor=east]{$\cC_n$}(0,0,1); 
\draw[black, thick] (0,1,1) -- (0,1,0) -- (1,1,0) -- (1,1,1); 
\draw[black, thick] (1,0,1) -- (1,0,0) -- (1,1,0) -- (1,1,1); 
\draw[gray, dashed] (0,0,0) -- (0,0,1);
\draw[gray, dashed] (0,0,0) -- (0,1,0);
\draw[gray, dashed] (0,0,0) -- (1,0,0);

\draw[blue, thick] (0.1,0,0) -- (0.6,0,1) -- (0.9,1,1) -- (0.4,1,0) -- node[anchor=east]{$\mathcal{H}$} (0.1,0,0); 

\node (p*1) at (0.6,0,1) {} ;
\node (p*2) at (0.4,1,0) {} ;
\node (p*3) at (0.9,1,1) {} ;
\filldraw[red] (p*1) circle (.3pt) node[anchor=north]{$\bp^{(1)}$};
\filldraw[red] (p*2) circle (.3pt) node[anchor=south]{$\bp^{(2)}$};
\filldraw[red] (p*3) circle (.3pt) node[anchor=south west]{$\bp^{(3)}\textcolor{blue}{=\,\bar \bp^{(2)}}$};

\node (p) at (0.625,0.5,0.75) {} ;
\filldraw[blue] (p) circle (.3pt) node[anchor=east]{$\bp$};

\node (hatp1) at (0.65,1,0.5) {}; 
\filldraw[blue] (hatp1) circle (.3pt) node[anchor=west]{$\bar \bp^{(1)}$};

\draw[red, dashed] (p*1) -- (p);
\draw[red, dashed] (p) -- (hatp1);

\end{tikzpicture}
\end{center}
\caption{Starting from $\bp$, the first step computes a first extreme point $\bp^{(1)}$. Then $\bp$ is written as a convex combination of $\bp^{(1)}$ and $\bar{\bp}^{(1)}$. The latter is decomposed (iteratively) into $\bp^{(2)}$ and then $\bp^{(3)}$. }
\label{FIG:carat}
\end{figure}
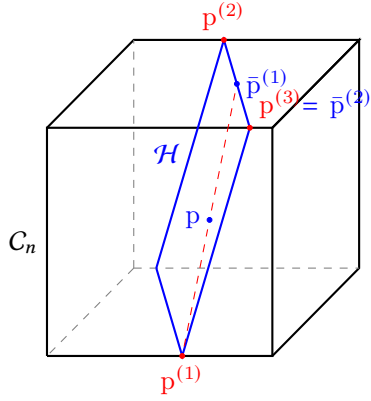

\subsection{Proof of Lemma \ref{Lemma:hypercube} and Algorithm \DECOMPOSE}
\label{SE:ProofLemma:hypercube}

\begin{center}
    \begin{algorithm}
\SetAlgoLined
\KwData{vector $\mathbf{b}=(b_i)_{i \in [n]} \in  \R^n$, budget $T \in \R_+$, vector $x \in \R^n$\;}
\KwResult{2 vectors $y,\bar{y}$, convex weight $\theta$\;}
 $t \leftarrow0$; $k \leftarrow0$, $z \leftarrow x\odot\mathds{1}_{\overline{\supp(x)}}$\Comment{Initialization}
 
\While{$t<T$}{
$k \leftarrow k+1$; $ z \leftarrow z+\mathds{1}_{\{k\}}\odot\mathds{1}_{\supp(x)}$; $t \leftarrow z^\top \mathbf{b}$\Comment{Find the index with non $\{0,1\}$ component}}
$\varepsilon' \leftarrow \frac{t-T}{b_k}$; $y\leftarrow z-\varepsilon'\mathds{1}_{\{k\}}$\Comment{Compute explicitly the extreme point $y$}

$\delta\leftarrow \min_{i \in \mathcal{S}}\frac{\mathds{1}\{x_i \geq y_i\}-x_i}{x_i-y_i}$; 
$\theta \leftarrow \frac{\delta}{1+\delta}$\;
$\bar{y}\leftarrow x+\eta(x-y)$\;
 \textbf{Output:} $(y,\bar{y},\theta)$ \Comment{(two vectors, one convex weight)}
 \caption{\Cube
 \label{Algo:ProjCube}}
\end{algorithm}
\end{center}

\begin{proof}
Assume that the coordinates of $\ell$ are sorted as $\ell_1 \leq \ell_2 \leq \ldots \leq \ell_n$.
For any subset $\mathcal{S} \subset [n]$, we denote its associated indicator vector by  $\mathds{1}_{\mathcal{S}}=(\mathds{1}\{i \in \mathcal{S}\})_i \in\{0,1\}^n$.

First, notice that there exists some $\eps_k \in [0,1]$ such that the following vector $y$ belongs to $\cC_n\cap \cH$:
$$y:=x\odot\mathds{1}_{\overline{\supp}(x)}+(\underbrace{1,1,\ldots,1}_{k-1 \text{ terms}},\eps_k,0,\ldots,0)\odot\mathds{1}_{\supp(x)},$$
 where $\overline{\supp}(x)=[n]\backslash \supp(x)$ is the complement of the support of $x$. The reason is simply that $$x\odot\mathds{1}_{\overline{\supp}(x)} \leq x \leq x\odot\mathds{1}_{\overline{\supp}(x)} + (1,\ldots,1)\odot\mathds{1}_{\supp(x)}\ ,$$
where the inequalities are component wise.

We now introduce 
$\bar{y}= x+\delta(x-y),$
where $\delta \in \R_+$ is defined by
$$
\delta= \delta_{i^*}:=\min_{i \in \supp(x)} \delta_i, \quad \text{with}\ \delta_i:=\frac{\mathds{1}\{x_{i}\geq y_{i}\}-x_{i}}{x_{i}-y_{i}}\ . 
$$
Then $\bar{y}$ belongs to $\cH$ as  $x$ and $y$ are two vectors of this affine hyperplane. Moreover, as $\eta_{i^*}=\min \eta_i$, it also holds that  $\bar{y} \in [0,1]^n$ and  
 $\bar{y}_{i^*}=\mathds{1}\{x_{i^*}\geq y_{i*}\}$, and therefore $\supp(\bar{y}) \subset \supp(x) \backslash\{i^*\}$. Finally, one just needs to define $\theta = \delta/(1+\delta)$. 

The construction of those quantities requires sorting the coordinates of $\ell$, finding the $k$-th coordinate of $y$, say, by binary search (computing the value of $\eps_k$ is immediate as the value of $y^\top \ell$ is linear in $p_k$) and finding $i^*$.
\end{proof}

\subsection{Proof of Lemma \ref{LM:reduc_to_four} and Algorithm \REDUCTION}\label{SE:ProofLM:reduc_to_four}

\begin{center}
    \begin{algorithm}
\SetAlgoLined
\KwData{length vector $\mathbf{b} \in  \R_+^n$, weights $\bp \in [0,1]^n$\;}
\KwResult{5 sets of indices\;}
$\mathcal{I}_{0}\leftarrow\emptyset$,  $\mathcal{I}_1 \leftarrow [n]$, $\mathcal{I}_2 \leftarrow \emptyset$, $\mathcal{I}_3 \leftarrow \emptyset$, $\mathcal{I}_{4}\leftarrow\emptyset$\Comment{Initialization of sets}

 \For{$i \in [n]$
 } 
 {
 \eIf{$p_i=1$}
 {$\mathcal{I}_{0}\leftarrow\mathcal{I}_{0}\cup\{i\}$, 
 $\mathcal{I}_{1}\leftarrow\mathcal{I}_{1}\backslash\{i\}$ \Comment{Dirac masses not treated, index removed}}
 {\eIf{$p_i\in(0,1)$}{$\mathcal{I}_{n}\leftarrow \{i\}$,  $\mathcal{I}_{1}\leftarrow\mathcal{I}_{1}\backslash\{i\}$\Comment{At most one strict mixture}}
 {$\mathcal{I}(i)\leftarrow\{i\}$ \Comment{End of initialization}}
 }}
 \While{$|\mathcal{I}_1|>3$}{
 $i_1\leftarrow \arg\min_{i \in \mathcal{I}_1} b_i$, $\mathcal{I}_{1}\leftarrow\mathcal{I}_{1}\backslash\{i_1\}$, 
 $i_2\leftarrow \arg\min_{i \in \mathcal{I}_1} b_i$ \Comment{$i_1,i_2$: 2 smallest uniform lengths}
 $b_{i_2} \leftarrow b_{i_2}+b_{i_1}$, 
 $\mathcal{I}(i_2)\leftarrow\mathcal{I}(i_2)\cup\mathcal{I}(i_1)$\Comment{2 uniforms combined in 1}
    }
$m=|\mathcal{I}_1|$, denote $\mathcal{I}_1=\{i^*_1,\ldots,i^*_m\}$\;
\For{$j \in [m]$}
{$\mathcal{I}_j \leftarrow \mathcal{I}(i^*_j)$\Comment{(At most) 3 sets of combined indices}}
 \textbf{Output:} $(\mathcal{I}_{0},\mathcal{I}_1,\mathcal{I}_2,\mathcal{I}_3,\mathcal{I}_4)$ \Comment{(Dirac, 3 sets of uniforms - possibly empty-, a strict mixture)}
 
 \caption{\REDUCTION\label{Algo:Reduc}}
\end{algorithm}
\end{center}

\begin{proof}
Write for simplicity $X_i=X_{A,i}$ and 
$$
X_i \sim (1-p_i)\delta_0+ p_i {\sf Unif}[0, b_i]\,,
$$
where $p_i \in \{0,1\}$ for $i \in [n-1]$, $p_n \in [0,1]$, and 
$b_i= ({\gamma^*v_{B,i}\wedge v_{A,i}})/{\lambda^*}$. If $p_i=0$, then $X_i$ is almost surely equal to 0 and this random variable can be removed from the analysis; we might therefore assume that $p_i >0$ for all $i$.

The joint mixability condition~\eqref{EQ:JMcond1} rewrites as 
\begin{equation}
    \label{EQ:JMcond1rep}
    2\max_k b_k \leq \sum_{i=1}^{n-1}b_i + p_nb_n  
\end{equation}
Without loss of generality, assume that  $b_1 \leq b_2 \leq \ldots\leq b_{n-1}$ and define $X_{12} \sim {\sf Unif}[0, b_1+b_2]$. Note that  $\{X_1,X_2,-X_{12}\}$ are jointly mixable (for instance, consider $X_1=\frac{b_1}{b_1+b_2}X_{12}$ and $X_2=\frac{b_2}{b_1+b_2}X_{12}$), and $\{X_{12}, X_3,\ldots,X_n\}$ are also jointly mixable if
\begin{equation}
    \label{EQ:JMcond1iter}
    2\max\{ \max_{k \geq 3} b_k, b_1+b_2\} \leq \sum_{i=1}^{n-1} b_i + p_nb_n  
\end{equation}
On the one hand, if $b_1+b_2 \leq b_{n-1}$ then~\eqref{EQ:JMcond1iter} follows readily from~\eqref{EQ:JMcond1rep}. On the other hand, if $b_1+b_2 >  b_{n-1}$ and $n-1 \geq 4$ then 
$$
2(b_1+b_2) \leq   b_1+b_2 + b_3+b_4  \leq b_1+b_2 + \sum_{k= 3}^{n-1} b_k +t^{(n)}b_n 
$$
so that~\eqref{EQ:JMcond1rep} is satisfied. The result follows from sequentially iterating this construction. Indeed, recall that $X_{12}$ is a variable created by mixing together $X_1$ and $X_2$. At a further step, this variable might be mixed again with another one, ending up with a set of, at least, 3 variables mixed together. At the end of the procedure, there are at most 3 remaining random variables $Y_j$, that are joint mix of (disjoints) subsets of $X_i$.  

We denote by $\mathcal{I}_j$ the set of indices of variables that are mixed to form $Y_j$ and, for any index $i \in \cI_j$, we just need to define $\theta_i=\frac{b_i}{\sum_{k \in \cI_j}b_k}$.
\end{proof}

\subsection{Algorithm \DISCRETIZE}\label{SE:DISCRETIZE}

\begin{center}
    \begin{algorithm}
\SetAlgoLined
\KwData{4 lengths $b^*_i \in \R_+$, positive weight $p_4 \in (0,1)$, discretization parameter $h$\;}
\KwResult{4 vectors $\mu_i \in \R_+^{d_i}$ with $d_i=\lfloor\frac{b^*_i}{h}\rfloor+1$\;}
\For{$i\in [3]$}{$\mu_i\leftarrow\frac{h}{b^*_i}\big(\underbrace{1,\ldots,1}_{\lfloor\frac{b^*_i}{h}\rfloor \text{times}},\frac{b^*_1}{h}-\lfloor\frac{h}{b^*_1}\rfloor\big)$}
$\mu_4\leftarrow p_4\frac{h}{b^*_i}\big(\frac{b^*_4}{h}-\lfloor\frac{h}{b^*_4}\rfloor,\underbrace{1,\ldots,1}_{\lfloor\frac{b^*_4}{h}\rfloor \text{times}}\big)+(1-p_4)(0,\ldots,0,1)$\;
\textbf{Output:} $(\mu_1,\mu_2,\mu_3,\mu_4)$ \Comment{(4 discrete ditributions)}
 
 \caption{\DISCRETIZE \label{Algo:Discretize}}
\end{algorithm}
\end{center}

\subsection{Algorithm \SINKHORN}\label{SE:SINKHORN}

\begin{center}
    \begin{algorithm}
\SetAlgoLined
\KwData{Marginals $\mu_j \in  \R^{d_j}, j=1, \ldots 4$\;
Precision level $\eta >0$\;}
\KwResult{A Coupling $\Gamma$, represented by $(\xi_1,\xi_2,\xi_3,\xi_4)\in \R^{d_1}\times\R^{d_2}\times\R^{d_3}\times\R^{d_4}$, with marginals such that $\sum_{i=1}^4 \|\bar \Gamma^{(i)} - \mu_i\|_1 \leq \eta$\;}
 $\Gamma\leftarrow {\bf 1} \in \R^{d_1 \times d_2\times d_3 \times 3}$\Comment{Initialization}
 
 $\xi_1\leftarrow {\bf 1} \in \R^{d_1}$; $\xi_2\leftarrow {\bf 1} \in \R^{d_2}$;
 $\xi_3\leftarrow {\bf 1} \in \R^{d_3}$;
 $\xi_4\leftarrow {\bf 1} \in \R^{d_4}$\;
 \While{$\sum_{i \in [4]} \|\mu_i-\Gamma_i\|_1> \eta$}{
  $\displaystyle \tau \leftarrow \argmax_{i\in [4]} \KL(\mu_i\|\Gamma_i)$\;
  $\displaystyle \xi_{\tau} \leftarrow \xi_{\tau}\odot \mu_{\tau}\oslash \bar \Gamma_{\tau}$\Comment{Matrix scaling}
  
$\Gamma_{ijke}=\xi_{1,i}\cdot \xi_{2,j}\cdot\xi_{3,k}\cdot\xi_{4,i+j+k+e}\,, \ \forall  i,j,k,e$\Comment{Computation of coupling}
      }
 \textbf{Output:} $(\xi_1,\xi_2,\xi_3,\xi_4)$
 \caption{\SINKHORN\label{Algo:Sinkhorn}}
 
\end{algorithm}
\end{center}
\newpage
\subsection{Algorithm \SAMPLE}\label{SE:SAMPLE}

\begin{center}
    \begin{algorithm}
\SetAlgoLined
\KwData{Coupling  $(\xi_1,\xi_2,\xi_3,\xi_4)$, budget $T\in\R_+$, length vector $\mathbf{b} \in \R_+^n$, sets of indices $(\cI_0,\cI_1,\cI_2,\cI_3,\cI_4)$, precision $h$ \;}
\KwResult{An allocation vector $X \in \R_+^n$\;}
Set $\{j_4\}=\cI_4$; $b_4^*\leftarrow b_{j_4}$\Comment{Initialization}

Sample $\tilde{Y}_1\sim\xi_1$, $\tilde{Y}_2\sim\xi_2$, $\tilde{Y}_3\sim\xi_3$\Comment{Sampling from computed coupling} 

Sample $\varepsilon =e \in \{0,1,2\}$ with proba.\  $\xi_{4,\tilde{Y}_1+\tilde{Y}_2+\tilde{Y}_3+e}$\; 

Sample $U \sim {\sf Unif}[0,1]$\Comment{smoothing}

\For{$i\in[3]$}{$b_i^*\leftarrow \sum_{j\in \cI_i}b_j$\;$Y'_i\leftarrow(\tilde{Y}_i+\frac{\varepsilon}{3}+\frac{U}{3})h\wedge b_i^*$}

$S\leftarrow Y'_1+Y'_2+Y'_3$\; $\zeta \leftarrow \mathds{1}\{S>T\} \frac{T-S}{3}+\mathds{1}\{S<T-b^*_4\} \frac{T-b^*_4-S}{3}$\;
\For{$i\in[3]$}{$\bar Y_i\leftarrow Y'_i+\zeta$\;\For{$j \in \cI_i$}{$X_j=\frac{b_j}{b^*_i}\bar Y_i$\Comment{Reconstruction of uniforms}}}
\For{$j \in \cI_0$}{$X_j=0$ \Comment{Dirac masses set at 0}}
$X_{j_4}=T-(\bar Y_1+\bar Y_2+\bar Y_3)$ \Comment{Budget saturation}

 \textbf{Output:} $X$
 \caption{\SAMPLE\label{Algo:Sample}}
\end{algorithm}
\end{center}

 \end{document}